\documentclass[11pt]{article}

\usepackage{fullpage}
\usepackage{amsmath}
\usepackage{amssymb}
\usepackage{amsthm}
\usepackage{graphicx}
\usepackage{psfrag}
\usepackage{verbatim}
\usepackage{url}
\usepackage{float}
\usepackage{color}

\def\E{\mathop{\mathbb{E}}}
\def\R{\mathbb{R}}

\def\eps{\epsilon}
\def\del{\delta}

\newcommand{\hide}[1]{}

\def\1{\mathbf{1}}

\def\lam {\lambda}

\def\var{\text{var}}

\newcommand{\on}{\{\pm 1\}}

\newtheorem{thm}{Theorem}
\newtheorem{lem}{Lemma}

\newtheorem{defn}{Definition}

\newtheorem{prop}{Proposition}

\usepackage{caption}
\usepackage{subcaption}

\begin{document}

\title{Subsampled Power Iteration: a Unified Algorithm for Block Models and Planted CSP's}
\author{Vitaly Feldman\thanks{IBM Research - Almaden}
\and Will Perkins\thanks{University of Birmingham and the IMA.}
\and Santosh Vempala\thanks{Georgia Tech. Supported in part by NSF award CCF-1217793.}
}
\date{}
\maketitle

\begin{abstract}
We present an algorithm for recovering planted solutions in two well-known models, the stochastic block model and planted constraint satisfaction problems, via a common generalization in terms of random bipartite graphs. Our algorithm matches up to a constant factor the best-known bounds  for the number of edges (or constraints) needed for perfect recovery and its running time is linear in the number of edges used. The time complexity is significantly better than both spectral and SDP-based approaches.

The main contribution of the algorithm is in the case of unequal sizes in the bipartition that arises in our reduction from the planted CSP.  Here our algorithm succeeds at a significantly lower density than the spectral approaches, surpassing a barrier based on the spectral norm of a random matrix.

Other significant features of the algorithm and analysis include (i) the critical use of power iteration with subsampling, which might be of independent interest; its analysis requires keeping track of multiple norms of an evolving solution (ii) the algorithm can be implemented statistically, i.e., with very limited access to the input distribution (iii) the algorithm is extremely simple to implement and runs in linear time, and thus is practical even for very large instances.
\end{abstract}

\thispagestyle{empty}
\newpage

\setcounter{page}{1}

\section{Introduction}

Partitioning a graph into parts based on the density of the edges within and between the parts is a fundamental algorithmic task both in its own right as a method of clustering data into similar pieces, and as a powerful subroutine of divide-and-conquer algorithms.  There are many choices for the number of parts required and the measure of the quality of a partition, and different choices give rise to algorithmic problems such as Max Clique, Max Cut, Uniform Sparsest Cut, and Min Bisection.

While finding an optimal graph partition is often an NP-hard problem in the worst case, the average-case study of graph partitioning problems is particularly rich, as the underlying distributions come from natural and widely studied models of random graphs (we review the previous work in Section~\ref{sec:previous}).

The simplest model is the stochastic block model: partition a set of vertices into two equal parts $A$ and $B$, and add edges independently, with probability $p$ for an edge within a part, and $q \ne p$ for a crossing edge.  The algorithmic task is to recover the partition given the random graph.  Generalizations include parts of unequal size, more than two parts, and more than two edge probabilities.

Another broad and fundamental class of algorithmic problems is the class of boolean Constraint Satisfaction Problems (CSP's, defined precisely below).  The average-case complexity of $k$-CSP's is a large area of research that intersects cryptography, computational complexity, probabilistic combinatorics and statistical physics.  In the planted $k$-SAT problem each constraint is a disjunction of $k$ literals, variables or their negations, eg. $\{ \overline x_5, x_6, \overline x_{10} \}$ and is referred to as $k$-clause.
A random instance of this problem is produced by choosing a random and uniform assignment $\sigma$ and then selecting $k$-clauses at random independently (but not necessarily uniformly) from the set of $k$-clauses satisfied by $\sigma$.  This distribution is guaranteed to have at least one satisfying assignment, $\sigma$. In the `noisy' version of the problem unsatisfied clauses are also included with some probability.
The algorithmic task is to recover the planted assignment $\sigma$.  An additional model of planted CSP's we consider is Goldreich's pseudorandom generator \cite{goldreich2000candidate} that has been studied in cryptography. We describe it in more detail below.

\subsection{Our results and techniques}
\label{sec:ourResults}

We propose a natural bipartite stochastic block model that generalizes the classic stochastic block model defined above. The key motivation for the study of this model is that the two types of planted $k$-CSP's  can be reduced to our block model, thus unifying graph partitioning and planted CSP's into one problem. We then give an algorithm for solving random instances of the model.

The model begins with two vertex sets, $V_1$ and $V_2$ (of possibly unequal size), each with a balanced partition, $(A_1, B_1)$ and $(A_2, B_2)$ respectively.  Edges are added independently at random between $V_1$ and $V_2$ with probabilities that depend on which parts the endpoints are in: edges between $A_1$ and $A_2$ or $B_1$ and $B_2$ are added with probability $\del p$, while the other edges are added with probability $(2-\del)p$, where $\del \in [0,2]$ and $p$ is the overall edge density. To obtain the stochastic block model we can identify $V_1$ and $V_2$.  To reduce planted CSP's to this model, we first reduce the problem to an instance of noisy $r$-XOR-SAT, where $r$ is the complexity parameter of the planted CSP distribution defined in \cite{feldman2013complexity} (see Sec.~\ref{sec:plantedCSPsec} for details). We then identify $V_1$ with literals, and $V_2$ with $(r-1)$-tuples of literals, and add an edge between literal $l \in V_1$ and tuple $t \in V_2$ when the $r$-clause consisting of their union appears in the formula.  The reduction leads to a bipartition with  $V_2$  much larger than $V_1$.

Our algorithm is based on applying power iteration with a sequence of matrices subsampled from the original adjacency matrix. This is in contrast to previous algorithms that compute the eigenvectors (or singular vectors) of the full adjacency matrix. Our algorithm has several advantages.  Such an algorithm, for the special case of square matrices, was previously proposed and analyzed in a different context by Korada et al \cite{korada2011gossip}.

\begin{itemize}
\item Up to a constant factor, the algorithm matches the best-known (and in some cases the best-possible) edge or constraint density needed for complete recovery of the planted partition or assignment. The algorithm for planted CSP's finds the planted assignment using $O(n^{r/2} \cdot \log n)$ clauses for a clause distribution of complexity $r$ (see Sec.~\ref{sec:plantedCSPsec} for the formal definition),
        nearly matching computational lower bounds for SDP hierarchies \cite{o2013goldreich} and the class of statistical algorithms \cite{feldman2013complexity}.
\item The algorithm is fast, running in time linear in the number of edges or constraints used, unlike other approaches that require computing eigenvectors or solving semi-definite programs.
\item The algorithm is conceptually simple and very easy to describe and implement.  In fact it can be implemented in the statistical query model, with very limited access to the input graph \cite{feldman2013complexity}. 
\item It is based on the idea of iteration with subsampling which may have further applications in the design and analysis of algorithms.

\item Most notably, the algorithm succeeds where generic spectral approaches fail.  For the case of the planted CSP, when $|V_2| \gg |V_1|$, our algorithm succeeds at a polynomial factor sparser density than the approaches of McSherry \cite{mcsherry2001spectral}, Coja-Oghlan \cite{coja2010graph}, and Vu \cite{vu2014simple}. The algorithm succeeds despite the fact that the `energy' of the planted vector with respect to the random adjacency matrix is far below the spectral norm of the matrix.  In previous analyses, this was believed to indicate failure of the spectral approach.  For a full discussion, see Section \ref{sec:spectral}.
\end{itemize}

\paragraph{}

The remainder of the paper is organized as follows:
\begin{itemize}
\item In Section \ref{sec:previous}, we review previous work.
\item  In Section \ref{sec:main} we formally define the model and present the main theorems.
\item In Section \ref{sec:algorithms} we describe the algorithm and analyze its performance.
\item In Section \ref{sec:CSPreduce} we give the reduction of the planted $k$-CSP problems to the bipartite stochastic block model.
\item  In Section \ref{sec:spectral} we compare our algorithm to other spectral approaches.
\item In Section \ref{sec:analysis} we present full details of the analysis.
\end{itemize}


\subsection{Related work}
\label{sec:previous}
\subsubsection*{Planted partitioning}



The stochastic block model was introduced in \cite{holland1983stochastic}.  Boppana \cite{boppana1987eigenvalues} gave a spectral-based algorithm for the model, and Jerrum and Sorkin  \cite{jerrum1998metropolis} gave a Metropolis approach. Dyer and Frieze \cite{dyer1989solution} and Blum and Spencer \cite{blum1995coloring} give algorithms for the related planted $k$-coloring model in which the vertex set is partitioned into $k$ equal parts and then edges crossing the partition are added independently at random while edges within the partition are forbidden.  Alon and Kahale \cite{alon1997spectral} gave a spectral algorithm for this problem.

 Later algorithms \cite{condon2001algorithms,feige2005spectral,coja2006spectral,bottcher2005coloring,coja2009finding}  improved either the running time or the density at which the algorithms succeed.  Of particular note is McSherry's algorithm \cite{mcsherry2001spectral} which is  based on a low-rank projection and is a generic algorithm for many planted partitioning problems, including the stochastic block model, the planted coloring problem, and the planted clique problem.  Coja-Oghlan \cite{coja2010graph} gave a refined general purpose partitioning algorithm and showed that the planted partition in the stochastic block model can be partially recovered when the average degree is just a constant.  Vu \cite{vu2014simple} recently gave a simple SVD-based general partitioning algorithm.

 While all of the above works seek to recover the partition at as low a density as possible, only recently have sharp thresholds for the possibility of recovery been identified.  Based on ideas from statistical physics, Decelle et al. \cite{decelle2011asymptotic} conjectured that in fact there is a sharp threshold for efficient recovery in the stochastic block model: if $p=a/n, q=b/n$, and $(a-b)^2 < 2(a+b)$ then any non-trivial recovery of the planted partition is impossible, while if $(a-b)^2 > 2(a+b)$ then there is an efficient algorithm (polynomial in the size of the graph) that gives a partition with significant correlation to the planting. Mossel, Sly, and Neeman proved the lower bound \cite{mossel2012stochastic}, and then Massoulie \cite{massoulie2014community} and Mossel, Neeman, Sly \cite{mossel2013proof} independently analyzed algorithms proving the upper bound. See also \cite{nadakuditi2012graph,krzakala2013spectral} for more on related algorithms.  Recent work has found algorithms that succeed at the optimal threshold for complete recovery \cite{abbe2014exact, mossel2014consistency}.

\subsubsection*{Planted $k$-CSP's}

A width-$k$ CSP is defined by a set of $m$ predicates denoted by $P_1,\ldots, P_m$ and a set of $m$ $k$-tuples of boolean variables from the set $V = \{x_1, \dots, x_n \}$ denoted by $C_1,\ldots,C_m$. Each predicate $P_i$ is a function from $\{ \pm 1 \}^k$ to $\{ \pm 1\}$. Identifying $+1$ with TRUE and $-1$ with FALSE, a predicate $P_i$ is satisfied by an assignment $\sigma : V \to \{ \pm 1\}$ if the evaluation of the predicate $P_i$ on the values assigned by $\sigma$ to the $k$-tuple of variables $C_i=(x_{i_1},\ldots,x_{i_k})$ is TRUE. Given such a $k$-CSP the algorithmic task is to find an assignment $\sigma$ that maximizes the number of satisfied constraints.

It was noted in \cite{barthel2002hiding} that drawing satisfied $k$-SAT clauses uniformly at random from all those satisfied by $\sigma$  does not result in a difficult algorithmic problem even if the number of observed clauses is relatively small (simply taking the majority vote for each variable suffices; see \cite{berthet2014optimal} for optimal statistical tests in this setting). However, by changing the proportions of clauses depending on the number of satisfied literals under $\sigma$, one can create a more challenging distribution over instances. Such `quiet plantings' were further studied in \cite{jia2005generating,achlioptas2005hiding,krzakala2009hiding,krzakala2014reweighted}.
Algorithms for solving instances with various values of relative proportions for planted 3-SAT were given in \cite{flaxman2003spectral,krivelevich2006solving,coja2010efficient}. Following \cite{feldman2013complexity}, we define such problems using a {\em planting distribution} $Q$. This distribution is defined over $\{ \pm 1 \}^k$ and for a vector $z$ it gives the proportion of clauses in which the values  $\sigma$ assigns to the $k$-tuple of literals in the clause is $z$ (see Section \ref{sec:plantedCSPsec} for the formal definition).

A related class of problems is one in which for some fixed predicate $P$, an instance is generated by choosing a planted assignment $\sigma$ uniformly at random and generating a set of $m$ random and uniform $P$-constraints. That is, each constraint is of the form
$P(x_{i_1},\ldots,x_{i_k}) = P(\sigma_{i_1},\ldots,\sigma_{i_k})$, where
$(x_{i_1},\ldots,x_{i_k})$ is a randomly and uniformly chosen $k$-tuple of variables (without repetitions). The algorithmic problem is to determine $\sigma$ given the $m$ $k$-tuples of variables and the corresponding values of $P$ on those tuples. Goldreich \cite{goldreich2000candidate} proposed a one-way function based on the apparent hardness of these problems. In his proposal the predicate is chosen randomly. The hardness of such problems for other predicates, most notably noisy $k$-XOR-SAT, has been used in cryptographic applications including public key cryptosystems \cite{Alekhnovich11a,applebaum2010public}, and secure two-party computation \cite{ishai2008cryptography}. It has also been used to derive hardness of approximation~\cite{applebaum2012pseudorandom} (for public discussions of these problems/assumptions see \cite{Barak:2012:Online, Trevisan:2007:Online}). Problems of this type are usually referred to as Goldreich's pseudorandom generator (PRG).



Bogdanov and Qiao \cite{bogdanov2009security} show that an SDP-based algorithm of Charikar and Wirth \cite{CharikarWirth:04} can be used to find the planted assignment for any predicate that is {\em not} pairwise-independent using $m=O(n)$ constraints. The same approach can be used to recover the input for any $t$-wise independent predicate using $O(n^{(t+1)/2})$ evaluations via the folklore birthday ``paradox"-based reduction to $t=1$ (see \cite{o2013goldreich} for details).

Finding the planted assignment in a randomly generated $k$-SAT formula is at least as hard as distinguishing between a satisfiable formula generated using a planted assignment and a randomly and uniformly generated $k$-SAT formula. Even this seemingly easier problem appears to be hard for certain planting distributions. This problem is a special case of another well-studied hard problem: refuting the satisfiability of SAT formulas in which the goal is to distinguish a satisfiable formula from a randomly an uniformly generated one (see \cite{feldman2013complexity} for the details of the connection).

It is important to note that in planted $k$-CSP's the planted assignment becomes identifiable with high probability after at most $O(n \log n)$ random clauses yet the best known efficient algorithms require $n^{\Omega(r/2)}$ clauses. Problems exhibiting this type of behavior have attracted significant interest in learning theory \cite{Blum:92a,DecaturGR99,Servedio:00jcss,Feldman:07jmlr,Shalev-ShwartzST:12,BerthetR:13,DanielyLS:13} and some of the recent hardness results are based on the conjectured computational hardness of the $k$-SAT refutation problem \cite{DanielyLS:13,DanielyS14}.

The connection of planted CSP's to graph partitioning is that many algorithms for planted CSP's use graph partitioning, and spectral graph partitioning in particular, as a subroutine.  Examples of such algorithms for some classes of constraint distributions include Flaxman's algorithm for planted $3$-SAT \cite{flaxman2003spectral}, Krivelevich and Vilenchik's algorithm \cite{krivelevich2006solving} that runs in expected polynomial time, and the algorithm of Coja-Oghlan, Cooper, Frieze \cite{coja2010efficient} for planted $3$-SAT distributions that include the quiet plantings described above. Many of the same spectral techniques have been applied here as well for the SAT refutation problem \cite{goerdt2001efficient,goerdt2003recognizing,coja2004strong,feige2004easily,friedman2005recognizing,coja2004techniques}.

\subsubsection*{Comparison with previous work}

The algorithm of Mossel, Neeman, and Sly \cite{mossel2013proof} for the case $n_1 = n_2$ also runs in near linear time, while other known algorithmic approaches for planted partitioning that succeed near the optimal edge density \cite{mcsherry2001spectral,coja2010graph,massoulie2014community} perform eigenvector or singular vector computations and thus require superlinear time, though a careful randomized implementation of low-rank approximations can reduce the running time of McSherry's algorithm substantially \cite{achlioptas2001fast}.

For planted satisfiability, the algorithm of Flaxman for planted $3$-SAT works for a subset of planted distributions (those with distribution complexity at most $2$ in our definition below) using $O(n)$ constraints, while the algorithm of Coja-Oghlan, Cooper, and Frieze \cite{coja2010efficient} works for planted $3$-SAT distributions that exclude unsatisfied clauses and uses $O(n^{3/2} \ln^{10} n ) $ constraints.

The only previous algorithm that finds the planted assignment in Goldreich's PRG for all predicates is the SDP-based algorithm of Bogdanov and Qiao \cite{bogdanov2009security} with the folklore generalization to $r$-wise independent predicates ({\em cf.}~\cite{o2013goldreich}). Similar to our algorithm, it uses $\tilde{O}(n^{r/2})$ constraints. This algorithm effectively solves the noisy $r$-XOR-SAT instance and therefore can be also used to solve our general version of planted satisfiability using $\tilde{O}(n^{r/2})$ clauses (via the reduction in Section \ref{sec:CSPreduce}).
Notably for both this algorithm and ours, having a completely satisfying planted assignment plays no special role: the number of constraints required depends only on the distribution complexity.

To the best of our knowledge, our algorithm is the first for the planted $k$-SAT problem that runs in linear time in the number of constraints used.

Our algorithm is arguably simpler than the approach in \cite{bogdanov2009security} and substantially improves the running time even for small $k$. Another advantage of our approach is that it can be implemented using restricted access to the distribution of constraints referred to as statistical queries \cite{kearns1998efficient,feldman2013statistical}. Roughly speaking, for the planted SAT problem this access allows an algorithm to evaluate multi-valued functions of a single clause on randomly drawn clauses or to estimate expectations of such functions, {\em without direct access to the clauses themselves}. Recently, in \cite{feldman2013complexity}, lower bounds on the number of clauses necessary for a polynomial-time statistical algorithm to solve planted $k$-CSPs were proved. It is therefore important to understand the power of such algorithms for solving planted $k$-CSPs. A statistical implementation of our algorithm gives an upper bound that nearly matches the lower bound for the problem. See \cite{feldman2013complexity} for the formal details of the model and statistical implementation.

 Korada, Montanari, and Oh \cite{korada2011gossip}  analyzed the `Gossip PCA' algorithm, which for the special case of an equal bipartition is the same as our subsampled power iteration.  The assumptions, model, and motivation in the two papers are different and the results incomparable. In particular, while our focus and motivation are on general (nonsquare) matrices, their work considers extracting a planting of rank $k$ greater than $1$ in the square setting.  Their results also assume an initial vector with non-trivial correlation with the planted vector.  The nature of the guarantees is also different.

Two other algorithms are similar in spirit to our approach: clustering via matrix powering of Zhou and Woodruff \cite{zhou2004clustering} and `Power Iteration Clustering' of Lin and Cohen \cite{lin2010power}.  In each, partitioning is performed by multiplying an initial vector by the adjacency matrix of the random graph repeatedly.  These methods are similar to ours in their simplicity; the subsampling in our algorithm allows us to carry out a rigorous analysis through many more iterations.


\section{Model and results}
\label{sec:main}


\subsubsection*{Bipartite stochastic block model}
\label{sec:generalBlock}

\begin{defn}
\label{def:bsbm}
For $\del \in [0,2] \setminus \{1\}$, $n_1, n_2$ even, and $\mathcal P_1 = (A_1, B_1)$, $\mathcal P_2 = (A_2, B_2)$ bipartitions of vertex sets $V_1, V_2$ of size $n_1, n_2$ respectively, we define the bipartite stochastic block model $B(n_1, n_2, \mathcal P_1, \mathcal P_2, \del,p)$ to be the random graph in which edges between vertices in $A_1$ and $A_2$ and $B_1$ and $B_2$ are added independently with probability $\del p$ and edges between vertices in $A_1$ and $B_2$ and $B_1$ and $A_2$ with probability $(2- \del) p$.
\end{defn}

\begin{figure}[htb]
        \center{\includegraphics[width=300pt]
        {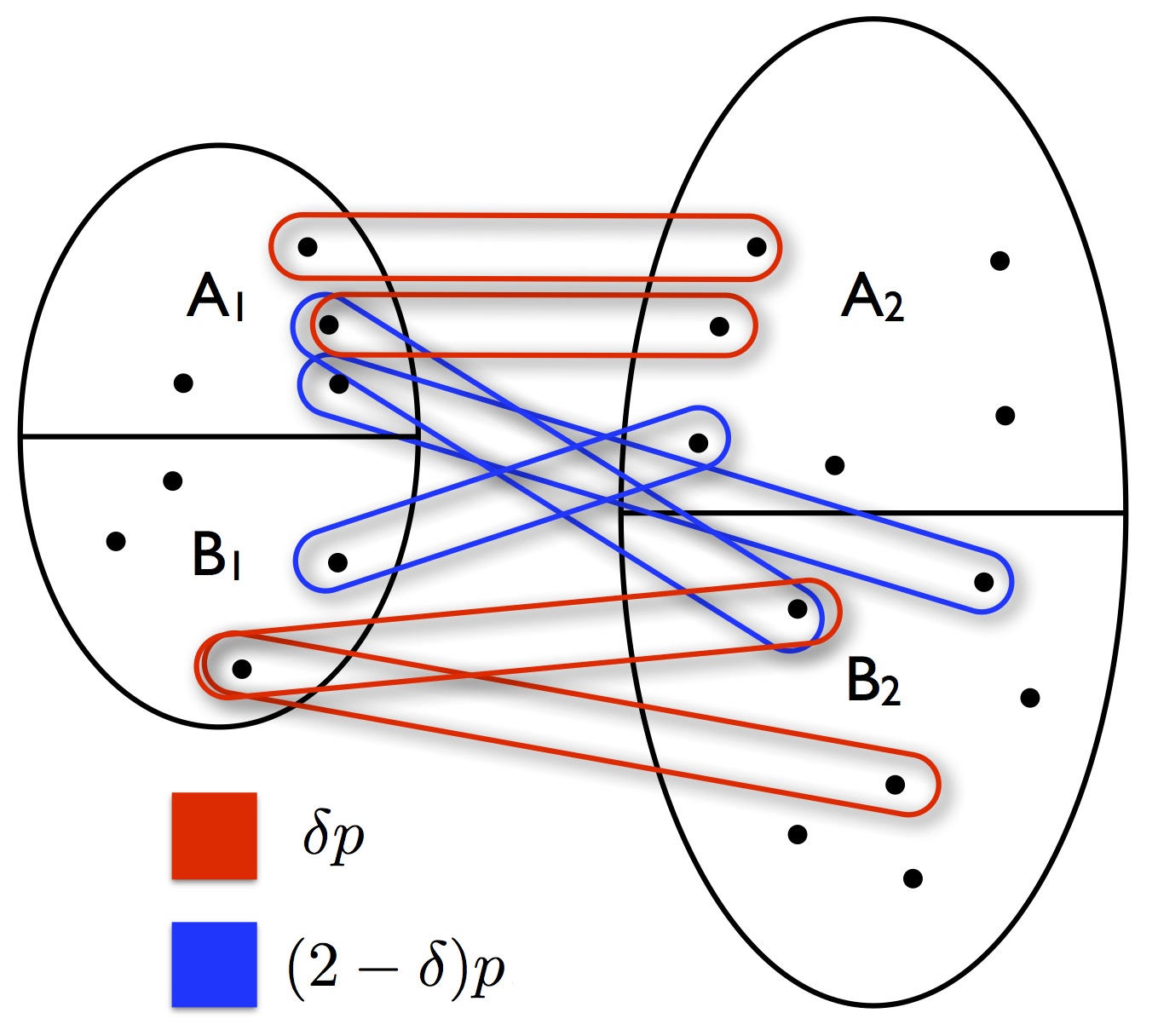}}
        \caption{Bipartite stochastic block model.}
      \end{figure}

Here $\del$ is a fixed constant while $p$ will tend to $0$ as $n_1, n_2 \to \infty$. Note that setting $n_1 = n_2 =n$, and identifying $A_1$ and $A_2$ and $B_1$ and $B_2$ gives the usual stochastic block model (with loops allowed); for edge probabilities $a/n$ and $b/n$, we have $\del = 2a/(a+b)$ and $p =(a+b)/2n $, the overall edge density. For our application to $k$-CSP's, it will be crucial to allow vertex sets of very different sizes, i.e. $n_2 \gg n_1$.

The algorithmic task for the bipartite block model is to recover one or both partitions (completely or partially) using as few edges and as little computational time as possible.  In this work we will assume that $n_1 \le n_2$, and we will be concerned with the algorithmic task of recovering the partition $\mathcal P_1$ completely, as this will allow us to solve the planted $k$-CSP problems described below.  We define complete recovery of $\mathcal P_1$ as finding the exact partition with high probability over the randomness in the graph and in the algorithm.  

\begin{thm}
\label{thm:block}
Assume $n_1 \le n_2$.  There is a constant $C$ so that the Subsampled Power Iteration algorithm described below completely recovers the partition $\mathcal P_1$ in the bipartite stochastic block model $B(n_1, n_2, \mathcal P_1, \mathcal P_2, \del,p)$ with probability $1-o(1)$ as $n_1 \to \infty$
when $p \ge \frac{C \log n_1}{(\delta-1)^2\sqrt{n_1 n_2}}$.
Its running time is $O\left(\sqrt{ n_1 n_2}\cdot \frac{\log n_1}{(\delta-1)^2}\right)$.
\end{thm}

Note that for the usual stochastic block model this gives an algorithm using $O(n \log n)$ edges and $O(n \log n )$ time, which is the best possible for complete recovery since that many edges are needed for every vertex to appear in at least edge.  With edge probabilities $a \log n/n$ and $b \log n/n$, our results requires $(a-b)^2 \ge C (a+b)$ for some absolute constant $C$, matching the dependence on $a$ and $b$ in \cite{boppana1987eigenvalues,mcsherry2001spectral} (see \cite{abbe2014exact} for a discussion of the best possible threshold for complete recovery).

 For any $n_1, n_2$, at least $\sqrt{n_1 n_2} $ edges are necessary for even non-trivial partial recovery, as below that threshold the graph consists only of small components (and even if a correct partition is found on each component, correlating the partitions of different components is impossible).  Similarly at least $\Omega( \sqrt{n_1 n_2} \log n_1)$ are needed for complete recover of $\mathcal P_1$ since below that density, there are vertices in $V_1$ joined only to vertices of degree $1$ in $V_2$.

For very lopsided graphs, with $n_2 \gg n_1 \log^2 n_1$, the running time is sublinear in the size of $V_2$; this requires careful implementation and is essential to achieving the running time bounds for planted CSP's described below.

\subsubsection*{Planted $k$-CSP's}
\label{sec:plantedCSPsec}
We now describe a general model for planted satisfiability problems introduced in \cite{feldman2013complexity}. For an integer $k$, let $\mathcal C_k$ be the set of all ordered $k$-tuples of literals from $x_1, \dots ,x_n, \overline x_1, \dots, \overline x_n$ with no repetition of variables. For a $k$-tuple of literals $C$ and an assignment $\sigma$,  $\sigma(C)$ denotes the vector of values that $\sigma$ assigns to the literals in $C$. A planting distribution $Q : \{\pm 1\}^k \to [0,1]$ is a probability distribution over $\{\pm 1\}^k$.
\begin{defn}
\label{def:csp}
Given a planting distribution $Q: \{\pm 1\}^k \to [0,1]$, and an assignment $\sigma \in \{ \pm 1 \}^n$, we define the random constraint satisfaction problem  $F_{Q,\sigma}(n,m)$ by drawing $m$ $k$-clauses from $\mathcal C_k$ independently according to the distribution
\[ Q_\sigma (C) = \frac{  Q(\sigma(C)) }{ \sum_{C^\prime \in \mathcal C_k} Q(\sigma(C^\prime))} \]
where $\sigma(C)$ is the vector of values that $\sigma$ assigns to the $k$-tuple of literals comprising $C$.
\end{defn}

\begin{figure}
        \center{\includegraphics[width=300pt]
        {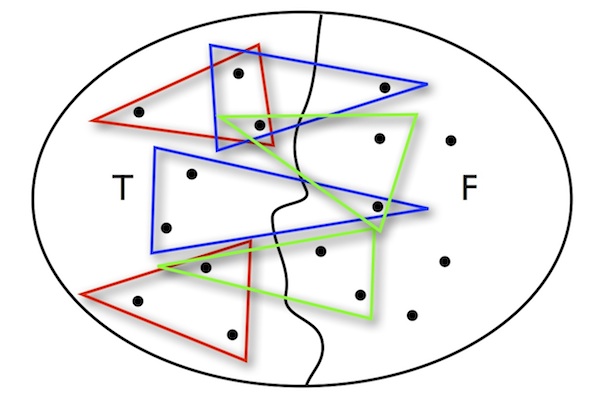}}
        \caption{Planted random 3-SAT.}
      \end{figure}

 \begin{defn}
 \label{def:distcomplex}
 The distribution complexity $r(Q)$ of the planting distribution $Q$ is the smallest integer $r \ge 1$ so that there is some $S \subseteq [k]$, $|S| =r$, so that the discrete Fourier coefficient $\hat Q(S)$ is non-zero.
 \end{defn}

In other words, the distribution complexity of $Q$ is $r$ if $Q$ is an $(r-1)$-wise independent distribution on $\{\pm 1\}^k$ but not an $r$-wise independent distribution. The uniform distribution over all clauses, $Q \equiv 2^{-k}$, has $\hat Q(S) =0$ for all $|S| \ge 1$, and so we define  its complexity to be $\infty$. The uniform distribution does not reveal any information about $\sigma$, and so inference is impossible. For any $Q$ that is not the uniform distribution over clauses, we have $1 \le r(Q) \le k$.

Note that the uniform distribution on $k$-SAT clauses with at least one satisfied literal under $\sigma$ has distribution complexity $r=1$. $r=1$ means that there is a bias towards either true or false literals.  In this case, a very simple algorithm is effective: for each variable, count the number of times it appears negated and not negated, and take the majority vote.  For distributions with complexity $r \ge 2$, the expected number of true and false literals in the random formula are equal and so this simple algorithm fails.

\begin{thm}
\label{thm:CSP}
For any planting distribution $Q$, there exists an algorithm that for any assignment $\sigma$, given an instance of $F_{Q,\sigma}(n,m)$ completely recovers the planted assignment $\sigma$ for $m=O(n^{r/2} \log n)$  using $O(n^{r/2} \log n)$ time, where $r \ge 2$ is the distribution complexity of $Q$.  For distribution complexity $r=1$, there is an algorithm that gives non-trivial partial recovery with $O(n^{1/2})$ constraints and complete recovery with $O( n \log n )$ constraints.
\end{thm}

We also show that the same result applies to recovering the planted assignment in Goldreich's PRG defined above.
\begin{thm}
For any predicate $P: \{\pm 1\}^k \to \{\pm 1\}$, there exists an algorithm that for any assignment $\sigma$, given $m$ random $P$-constraints completely recovers the planted assignment $\sigma$ for $m=O( n^{r/2} \log n)$ and using  $O(n^{r/2} \log n)$ time, where $r \ge 2$ is the degree of the lowest-degree non-zero Fourier coefficient of $P$. For $r=1$, the algorithm gives non-trivial partial recovery with $O(n^{1/2})$ constraints and complete recovery with $O( n \log n )$ constraints.
\end{thm}

\section{The algorithm}
\label{sec:algorithms}

We now present our algorithm for the bipartite stochastic block model.  We define vectors $u$ and $v$ of dimension $n_1$ and $n_2$ respectively, indexed by $V_1$ and $V_2$, with $u_i = 1$ for $i \in A_1$, $u_i=-1$ for $i \in B_1$, and similarly for $v$.  To recover the partition $\mathcal P_1$ it suffices to find either $u$ or $-u$.  We will find this vector by multiplying a random initial vector $x_0$ by a sequence of centered adjacency matrices and their transposes.

We form these matrices as follows: let $G_p$ be the random bipartite graph drawn from the model $B(n_1, n_2, \mathcal P_1, \mathcal P_2, \del,p)$, and $T$ a positive integer. Then  form $T$ different bipartite graphs $G_1, \dots ,G_T$ on the same vertex sets $V_1,V_2$ by placing each edge from $G_p$ uniformly and independently at random into one of the $T$ graphs. The resulting graphs have the same marginal distribution.

Next we form the $n_1 \times n_2$ adjacency matrices $A_1, \dots, A_T$ for $G_1, \dots G_T$ with rows indexed by $V_1$ and columns by $V_2$ with a $1$ in entry $(i,j)$ if vertex $i \in V_1$ is joined to vertex $j\in V_2$.  Finally we center the matrices by defining $M_i = A_i - \frac{p}{T} J$ where $J$ is the $n_1 \times n_2$ all ones matrix.
\hide{
}

  In the bipartite block model, these subsampled matrices are nearly independent (see Lemma \ref{OddCalcProp}), leading to a strong bound on the number of iterations required to solve the problem.  The subsampling also mitigates the influence of high-degree vertices leading to significant improvement over the spectral approach for a large subclass of planted CSP's.

 The analysis of the algorithm proceeds by tracking a potential function, $U_i = x^i \cdot u$ for a sequence of unit vectors $x^0, x^1, \dots$ of dimension $n_1$.  We must bound various norms of the $x^i$'s as well as norms of a sequence of auxiliary vectors $y^1, y^2, \dots$ of dimension $n_2$. We use superscripts to denote the current step of the iteration and subscripts for the components of the vectors, so $x^i_j$ is the $j$th coordinate of the vector after the $i$th iteration.

 The basic iterative steps are the multiplications $y=M^T x $ and $x=M y$.



\begin{figure}[H]
\begin{center}
\fbox{
\parbox{\textwidth}{
{\bf Algorithm: Subsampled Power Iteration.}
\begin{enumerate}
\item Form $ T =10 \log n_1$ matrices $M_1, \dots ,M_T$ by uniformly and independently assigning each edge of the bipartite block model to a graph $G_1, \ldots ,G_T$, then forming the matrices $M_i = A_i - \frac{p}{T} J$, where $A_i$ is the adjacency matrix of $G_i$ and $J$ is the all ones matrix.
\item Sample $x\in \{ \pm 1\}^{n_1}$ uniformly at random and let $x^0 = \frac{x}{\sqrt {n_1}}$.
\item For $i=1$ to $T/2 $ let
\[
y^{i} = \frac{ M^T_{2i-1} x^{i-1} }{\| M^T_{2i-1} x^{i-1} \|}; \quad x^{i} = \frac{ M_{2i} y^{i} }{\| M_{2i} y^{i} \|}; \quad
 z^{i} = \text{sgn}(x^{i}).
 \]
\item For each coordinate $j\in [n_1]$ take the majority vote of the signs of $z_j^i$ for all $i \in \{T/4, \ldots,
  T/2\}$ and call this vector $\overline v$:
\[ \overline v_j = \text {sgn} \left( \sum_{i=T/2}^{T} z^i_j \right)  .\]
\item Return the partition indicated by $\overline v$.
\end{enumerate}
}
}
\end{center}
\end{figure}

The analysis of the resampled power iteration algorithm proceeds in four phases,  during which we track the progress of two vectors $x^i$ and $y^i$, as measured by their inner product with $ u$ and $v$ respectively.  We define $U_i :=  u \cdot x^i$ and $V_i := v \cdot y^i$. Here we give an overview of each phase; the complete analysis is in Section \ref{sec:analysis}.

\begin{itemize}
\item \textbf{Phase 1.} Within $\log n_1$ iterations,  $|U_i|$ reaches $\log n_1 $.  We show that conditioned on the value of $U_i$, there is at least a $1/2$ chance that $|U_{i+1}| \ge 2| U_i|$; that $U_i$ never gets too small; and that in $\log n_1$ steps, a run of $\log \log n_1$ doublings pushes the magnitude of $U_i$ above $ \log n_1$.
\item \textbf{Phase 2.} After reaching $ \log n_1$, $|U_i|$ makes steady, predictable progress, doubling at each step whp until it reaches $\Theta(\sqrt{n_1})$, at which point we say $x^i$ has strong correlation with $u$.
\item \textbf{Phase 3.} Once $x^i$ is strongly correlated with $u$, we show that $z^{i+1}$ agrees with either $u$ or $-u$ on a large fraction of coordinates.
\item \textbf{Phase 4.} We show that taking the majority vote of the coordinate-by-coordinate signs of $z^i$ over $O(\log n_1)$ additional iterations gives complete recovery whp.
\end{itemize}

\subsubsection*{Running time}


If $n_2 = \Theta(n_1)$, then a straightforward implementation of the algorithm runs in time linear in the number of edges used: each  entry of $x^{i} = M y^i$ (resp. $y^{i} = M^T x^{i-1}$) can be computed as a sum over the edges in the graph associated with $M$. The rounding and majority vote are both linear in $n_1$.

However, if $n_2 \gg n_1 $, then simply initializing the vector $y^i$ will take too much time.  In this case, we have to implement the algorithm more carefully.

Say we have a vector $x^{i-1}$ and want to compute $x^{i} = M_{2i} y^i $ without storing the   vector $y^{i}$. 
Instead of computing  $y^{i} = M_{2i-1}^T x^{i-1}$, we create a set $S^{i} \subset V_2$ of all vertices with degree at least $1$ in the current graph $G_{2i-1}$ corresponding to the matrix $M_{2i-1}$.  The size of $S^{i}$ is bounded by the number of edges in $G_{2i-1}$, and checking membership can be done in constant time with a data structure of size $O(|S^{i}|)$ that requires expected time $O(|S^{i}|)$ to create \cite{fredman1984storing}.

Recall that $M_{2i-1}= A_{2i-1} - q J$.  Then we can write
\[ y^i = ( A_{2i-1} - q J)^T x^{i-1}  = \hat y - q \left( \sum_{j=1}^{n_1} x_j^{i-1} \right ) \mathbf{1}_{n_2} = \hat y - q L \mathbf{1}_{n_2}, \]
 where $\hat y$ is $0$ on coordinates $j \notin S^i$, $L=\sum_{j=1}^{n_1} x_j^{i-1}$,   and $\mathbf 1_{n_2}$ is the all ones vector of length $n_2$.  

 Then to compute $x^{i} = M_{2i} y^i$, we write

 \begin{align*}
 x^i = (A_{2i} - qJ) y^i &= (A_{2i} - qJ)( \hat y - q L \mathbf{1}_{n_2})  \\
 &=(A_{2i} - qJ) \hat y  - q L A_{2i} \mathbf{1}_{n_2} +  q^2 L  J \mathbf{1}_{n_2} \\
 &=  A_{2i}  \hat y - qJ \hat y  - q L A_{2i} \mathbf{1}_{n_2} + q^2L n_2 \mathbf{1}_{n_1}
 \end{align*}

 We bound the running time of the computation as follows:  we can compute $\hat y$ in linear time in the number of edges of $G_{2i-1}$ using $S^i$.  Given $\hat y$, computing $A_{2i}  \hat y $ is linear in the number of edges of $G_{2i}$ and computing $qJ \hat y$ is linear in the number of non-zero entries of $\hat y$, which is bounded by the number of edges of $G_{2i-1}$.  Computing $L = \sum_{j=1}^{n_1} x_j^{i-1}$ is linear in $n_1$ and gives  $ q^2L n_2 \mathbf{1}_{n_1}$. Computing $q L A_{2i} \mathbf{1}_{n_2}$ is linear in the number of edges of $G_{2i}$. All together this gives our linear time implementation.

\section{Reduction of planted $k$-CSP's to the block model}
\label{sec:CSPreduce}

Here we describe how solving the bipartite block model suffices to solve the planted $k$-CSP problems.

Consider a planted $k$-SAT problem $F_{Q,\sigma}(n,m)$ with distribution complexity $r$.  Let $S\subseteq [k]$, $|S| =r$, be such that $\hat{Q}(S) = \eta \neq 0$. Such an $S$ exists from the definition of the distribution complexity.   We assume that we know both $r$ and this set $S$, as trying all possibilities (smallest first) requires only a constant factor ($2^r$) more time.

 We will restrict each $k$-clause in the formula to an $r$-clause, by taking the $r$ literals specified by the set $S$. If the distribution $Q$ is known to be symmetric with respect to the order of the $k$-literals in each clause, or if clauses are given as unordered sets of literals, then we can simply sample a random set of $r$ literals (without replacement) from each clause.

  We will show that restricting to these $r$ literals from each $k$-clause induces a distribution on $r$-clauses defined by $Q^\delta: \{\pm 1\}^r \to \R^+$ of the form
$Q^\delta(C) = \del/2^{r}$ for $|C|$ even, $Q^\delta(C) = (2-\del)/2^{r}$ for $|C|$ odd,
for some $\del \in [0,2]$ , $\del \ne 1$, where $|C|$ is the number of TRUE literals in $C$ under $\sigma$. This reduction allows us to focus on algorithms for the specific case of a parity-based distribution on $r$-clauses with distribution complexity $r$.

Recall that for a  function $f:\{-1,1\}^k \rightarrow \R$, its Fourier coefficients are defined for each subset $S \subset [k]$ as
\begin{equation*}
\hat f(S) = \E_{x \sim \{-1,1\}^k}[f(x)\chi_S(x)]
\end{equation*}
where  $\chi_S$ are the Walsh basis functions of $\{ \pm 1 \}^k$ with respect to the uniform probability measure, i.e.,  $\chi_S(x) = \prod_{i \in S} x_i $.

\begin{lem}
If the function $Q: \on^k \to \R^+$ defines a distribution $Q_{\sigma}$ on $k$-clauses with distribution complexity $r$ and planted assignment $\sigma$, then for some $S \subseteq [k]$, $|S|=r$ and $\del \in [0,2] \setminus \{1 \}$, choosing $r$ literals with indices in $S$ from a clause drawn randomly from $Q_{\sigma}$ yields a random $r$-clause from $Q^\delta_\sigma$.
\end{lem}

\begin{proof}
From Definition \ref{def:distcomplex} we have that there exists an $S$ with $|S| = r$ such that $\hat Q(S) \ne 0$. Note that by definition,
\begin{align*}
\hat Q(S) = \E_{x \sim \{ \pm 1 \}^k} [Q(x) \chi_S(x) ]
&= \frac{1}{2^k} \sum_{x \in \{ \pm 1 \}^k}  Q(x) \chi_S(x) \\
&= \frac{1}{2^k} \left(  \sum_{x: \in \{ \pm 1 \}^k : x_S \text{ even} } Q(x) -  \sum_{x: \in \{ \pm 1 \}^k : x_S \text{ odd} } Q(x)   \right ) \\
&= \frac{1}{2^k} \left ( \Pr[x_S   \text{ even}]  - \Pr [x_S   \text{ odd}] \right )
\end{align*}
where $x_S$ is $x$ restricted to the coordinates in $S$, and so if we take $\del =1+ 2^k \hat Q(S)$, the distribution induced by restricting $k$-clauses to the $r$-clauses specified by $S$ is $Q^\delta_\sigma$. Note that by the definition of the distribution complexity, $\hat Q(T) = 0$ for any $1 \le |T| < r$, and so the original and induced distributions are uniform over any set of $r-1$ coordinates.

\end{proof}

First consider the case $r=1$.  Restricting each clause to $S$ for $|S| =1$, induces a noisy 1-XOR-SAT distribution in which a random true literal appears with probability $\delta$ and random false literal appears with probability $2-\delta$. The simple majority vote algorithm described above suffices: set each variable to $+1$ if it appears more often positively than negated in the restricted clauses of the formula; to $-1$ if it appears more often negated; and choose randomly if it appears equally often.   
Using $c \sqrt{t \log(1/\eps)}$ clauses for $c=O(1/|1-\delta|^2)$ this algorithm will give an assignment that agrees with $\sigma $ (or $-\sigma$) on $n/2 + t \sqrt{n}$ variables with probability at least $1- \eps$; using $c n \log n$ clauses it will recover $\sigma$ exactly with probability $1-o(1)$.

Now assume that $r\ge 2$. We describe how the parity distribution $Q_\sigma^\delta$ on $r$-constraints induces a bipartite block model.  Let $V_1$ be the set of $2n$ literals of the given variable set, and $V_2 $  the collection of all $( r-1)$-tuples of literals.  We have $n_1 = |V_1| = 2n$ and $n_2 = |V_2| =  \binom{2n}{r-1}$. We partition each set into two parts as follows:  $A_1 \subset V_1$ is the set of false literals under $\sigma$, and $B_1$ the set of true literals.  $A_2 \subset V_2$ is the set of $(r-1)$-tuples with an even number of true literals under $\sigma$, and $B_2$ the set of $(r-1)$-tuples with an odd number of true literals.

For each $r$-constraint $(l_1, l_2, \dots ,l_r)$, we add an edge in the block model between the tuples $l_1 \in V_1$ and $ (l_{2} ,\ldots, l_r) \in V_2$.  A constraint drawn according to $Q^\delta_\sigma$ induces a random edge between $A_1$ and $A_2$ or $B_1$ and $B_2$ with probability $\del/2$ and between $A_1$ and  $B_2$ or $B_1$ and $A_2$ with probability $1-\del/2$, exactly the distribution of a single edge in the bipartite block model.

Now the model in Definition \ref{def:csp} is that of $m$ clauses selected independently with replacement according to a given distribution, while in Definition \ref{def:bsbm}, each edge is present independently with a given probability.  To reduce from $m$ independent edges with replacement to the binomial model, we can fix some $\eps >0$ (e.g. $\eps =1/2$), draw a Poisson random variable $Z$ with mean $(1- \eps)m$, and select the first $Z$ of the $m$ edges (whp $Z \le m$), discarding any multiple edges.  By Poisson thinning, this leaves us with a graph where each edge $e$ appears independently with probability $p_e$, where $p_e = \Pr[ \text{Poisson}((1-\eps)m \cdot q_e) \ge 1]$ where $q_e$ is the probability of edge $e$ in the single edge distribution.  In particular, if for example $e$ joins a vertex in $A_1$ to a vertex in $A_2$ and $\eps =1/2$, then $q_e = \del/2\cdot \frac{2}{n_1 n_2}$ and
\[ p_e = 1 - \exp \left (-  \frac{  \del m} {2 n_1 n_2}  \right )  = \del p \]
where $p = \frac{m}{2n_1 n_2} (1+o(1))$.

Recovering the partition $\mathcal P_1 = A_1 \cup B_1$ in this bipartite block model partitions the  literals into true and false sets giving $\sigma$ (up to  sign).

The reduction from Goldreich's PRG to the bipartite block model is even simpler. By definition, the value of the predicate is correlated with the parity function of some $r$ of the $k$ inputs of the predicate (see for example \cite{bogdanov2009security}). Therefore the input can be seen as produced by the noisy $r$-XOR predicate on random and uniform $r$-tuples of variables. The $r$-tuples for which this predicate is equal to 1 give an instance of noisy $r$-XOR-SAT. A bipartite block model can now be formed on the set of variables and $(r-1)$-tuples of variables (instead of literals) analogously to the construction above.

The key feature of our bipartite block model algorithm is that it uses $\tilde O(\sqrt{n_1 n_2})$ edges (i.e. $p = \tilde O((n_1n_2)^{-1/2})$, corresponding to  $\tilde O(n^{r/2})$ clauses in the planted CSP.

\section{Comparison with spectral approach}
\label{sec:spectral}

As noted above, many approaches to graph partitioning problems and planted satisfiability problems use eigenvectors or singular vectors. These algorithms are essentially based on the signs of the top eigenvector of the centered adjacency matrix being correlated with the planted vector. This is fairly straightforward to establish when the average degree of the random graph is large enough.  However, in the stochastic block model, for example, when the average degree is a constant, vertices of large degree dominate the spectrum and the straightforward spectral approach fails (see \cite{krzakala2013spectral} for a discussion and references).

In the case of the usual block model, $n_1 =n_2 =n$, while our approach has a fast running time, it does not save on the number of edges required as compared to the standard spectral approach: both require $\Omega(n \log n)$ edges.  However, when $n_2 \gg n_1$, eg. $n_1 = \Theta(n), n_2=\Theta(n^{k-1})$ as in the case of the planted $k$-CSP's for odd $k$, this is no longer the case.

Consider the general-purpose partitioning algorithm of  \cite{mcsherry2001spectral}.  Let $G$ be the matrix of edge probabilities: $G_{ij}$ is the probability that the edge between vertices $i$ and $j$ is present.  Let $G_u, G_v$ denote columns of $G$ corresponding to vertices $u, v$.  Let $\sigma^2$ be an upper bound of the variance of an entry in the adjacency matrix, $s_m$ the size of the smallest part in the planted partition, $q$ the number of parts, $\delta $ the failure probability of the algorithm, and $c$ a universal constant. Then the condition for the success of McSherry's partitioning algorithm is:

\[ \min_{u, v \text{ in different parts}} \| G_u - G_v \|^2 > c q \sigma^2 (n/s_m + \log (n/\del)) \]

Similar conditions appear in \cite{coja2010graph,vu2014simple}.  In our case, we have $q=4$, $n= n_1 + n_2$, $s_m = n_1/2$, $\sigma^2 = \Theta(p)$, and $\|G_u - G_v\|^2 = 4(\del -1)^2 p^2 n_2$.  When $n_2 \gg n_1 \log n$, the condition requires $p =\Omega(1/n_1)$, while our algorithm succeeds when $p = \Omega( \log n_1/ \sqrt{n_1 n_2})$. In our application to planted CSP's with odd $k$ and $n_1 = 2n, n_2 = \binom {2n}{k-1}$, this gives a polynomial factor improvement.

 In fact, previous spectral approaches to planted CSP's or random $k$-SAT refutation worked for even $k$  using  $n^{k/2}$ constraints \cite{goerdt2001efficient, coja2003certifying, feige2005spectral}, while algorithms for odd $k$ only worked for $k=3$ and used considerably more complicated constructions and techniques \cite{feige2004easily, friedman2005recognizing, coja2010efficient}. In contrast to previous approaches, our algorithm unifies the algorithm for planted $k$-CSP's for odd and even $k$, works for odd $k >3$, and is particularly simple and fast.

 We now describe why previous approaches faced a spectral barrier for odd $k$, and how our algorithm surmounts it.

  The previous spectral algorithms for even $k$ constructed a similar graph to the one in the reduction above: vertices are $k/2$-tuples of literals, and with edges between two tuples if their union appears as a $k$-clause.  The distribution induced in this case is the stochastic block model.  For odd $k$, such a reduction is not possible, and one might try a bipartite graph, with either the reduction described above, or with $\lfloor k/2 \rfloor$-tuples and $\lceil k/2 \rceil$-tuples (our analysis works for this reduction as well).  However, with $\tilde O (k/2)$ clauses, the spectral approach of computing the largest or second largest singular vector of the adjacency matrix does not work.

Consider $M$ from the distribution $M(p)$.  Let $u$ be the $n_1$ dimensional vector indexed as the rows of $M$ whose entries are $1$ if the corresponding vertex is in $A_1$  and $-1$ otherwise.  Define the $n_2$ dimensional vector $v$ analogously.  The next propositions summarize  properties of $M$.

\begin{prop}
$\E(M)=(\del -1)puv^T$.
\end{prop}

\begin{prop}
\label{SVDprop}
Let $M_1$ be the rank-$1$ approximation of $M$ drawn from $ M(p)$.  Then
$\|M_1 - \E(M)\| \le 2\|M-\E(M)\|.$
\end{prop}

\begin{proof}
Using the triangle inequality and then the optimality of $M_1$, $\|M_1 - \E(M)\| \le \|M - \E(M)\| + \|M-M_1\| \le 2\|M-\E(M)\|$.
\end{proof}

The above propositions suffice to show high correlation between the top singular vector and the vector $u$ when $n_2 = \Theta( n_1) $ and $p = \Omega(\log n_1 / n_1)$. This is because the norm of $\E(M)$ is $p\sqrt{n_1n_2}$; this  is higher than $O(\sqrt{pn_2})$, the norm of $M-\E(M)$ for this range of $p$. Therefore the top singular vector of $M$ will be correlated with the top singular vector of $\E(M)$. The latter is a rank-$1$ matrix  with $u$ as its left singular vector.

However, when $n_2 \gg n_1$ (eg. $k$ odd) and $p = \tilde O((n_1 n_2)^{-1/2})$, the norm of the zero-mean matrix $M-\E(M)$ is in fact much larger than the norm of $\E(M)$.  Letting $x^{(i)}$ be the vector of length $n_1$ with a $1$ in the $i$th coordinate and zeroes elsewhere, we see that $\|Mx^{(i)}\|_2 \approx \sqrt{p n_2}$, and so
 $\| M- \E(M)\| = \Omega(\sqrt{p n_2})$, while $\| \E(M) \| = O(p\sqrt{n_1n_2})$; the former is $\Omega((n_2/n_1)^{1/4})$ while the latter is $O(1)$).
In other words, the top singular value of $M$ is much larger than the value obtained by the vector corresponding to the planted assignment! The picture is in fact richer: the straightforward spectral approach succeeds for $p \gg n_1^{-2/3} n_2^{-1/3}$, while  for $p \ll n_1^{-2/3}n_2^{-1/3}$, the top left singular vector of the centered adjacency matrix is asymptotically uncorrelated with the planted vector \cite{FlorescuP:15manu}. In spite of this, one can exploit correlations to recover the planted vector below this threshold with our resampling algorithm, which in this case provably outperforms the spectral algorithm.

\section{Analysis of the subsampled power iteration algorithm}
\label{sec:analysis}

We abuse notation and let $A_1, B_1, A_2, B_2$ denote the sets of coordinates of the corresponding vertex sets. Recall that $u \in \{ \pm 1 \}^{n_1}$ is $1$ on $A_1$ and $-1$ on $B_1$, and  $v \in \{ \pm 1 \}^{n_2}$ is $1$ on $A_2$, $-1$ on $B_2$. Set $T= 10 \log n_2$, $p = \frac{100 T}{(\del-1)^2 \sqrt{n_1 n_2}}$ and $q = p/T$. For convenience we denote $d = 100/(\del-1)^2$. We assume WLOG that $\del >1$.

 Recall that the sequence of matrices $M_1, \dots M_T$ is formed by taking $G_p$ and randomly assigning each edge to one of $T$ different bipartite graphs, then forming the corresponding centered adjacency matrices.   The marginal distribution of each $M_i$ is a random $n_1 \times n_2$ matrix with independent entries such that the entry $(i,j)$ takes value $1-q$ with probability $\del q$, $-q$ otherwise if $i \in A_1, j \in A_2$ or $i \in B_1, j \in B_2$, and value $1-q$ with probability $(2-\del) q$, $-q$ otherwise if $i \in A_1, j \in B_2$ or $i \in B_1, j \in A_2$.
 
The matrices are not independent, but are nearly independent.  Consider the distribution of $M_i$ conditioned on the matrices $M_1, \dots M_{i-1}$, call this set of edges. Let $E_{i-1}$ be the set of all edges from $G_p$ that are assigned to one of $M_1, M_2, \dots M_{i-1}$.  Conditioned on $M_1, \dots M_{i-1}$, the entries of $M_i$ are independent.  $M_i$ is necessarily $-q$ in every entry $(u,v)$ with $(u,v) \in E_{i-1}$.   All other entries take the values $1-q$ with probabilities 
\begin{align*}
\rho_i &:= \frac{ \Pr[ (u,v) \in G_i \wedge (u,v) \notin G_1, \dots G_{i-1}]}{\Pr[  (u,v) \notin G_1, \dots G_{i-1}]} \\
&=\frac{\del q}{1- \del p  \frac{i-1}{T} } = \del q + O( pq)  \quad \mbox{ if } u \in A_1, v\in A_2 \mbox{ or } u \in B_1, v \in B_2\\
& \text{ and }\\
\overline \rho_i &:= \frac{ \Pr[ (u,v) \in G_i \wedge (u,v) \notin G_1, \dots G_{i-1}]}{\Pr[  (u,v) \notin G_1, \dots G_{i-1}]} \\
&= \frac{ (2-\del)q }{1 - (2-\del) p  \frac{i-1}{T} } = (2-\del) q+ O( pq)  \quad \mbox{ if } u \in A_1, v\in B_2 \mbox{ or } u \in A_2, v \in B_1
\end{align*}
and the value $-q$ otherwise.  The deviation from the fully independent setting is the $O(pq)$ term.
 
Let $H_{i-1}$ be the event that (1) $| E_{i-1}| \le 2  n_1 n_2 p$, (2) each vertex of $V_1$ appears in $E_{i-1}$ at most $3 n_2 p$ times, and (3) each vertex of $V_2$ appears in $E_{i-1}$ at most $3 T (\del-1)^{-2}  $ times.  $H_{i-1}$ holds for all $1 \le i \le T$ whp from  simple Chernoff bounds.  We will condition on the set $E_{i-1}$ and the event $H_{i-1}$, to calculate the effect of multiplying a unit vector by $M_i$ or $M_i^T$. The calculations are based on bounding two deviations from the simpler calculations involving the marginal distribution of $M_i$: the deviations from the probabilities $\rho_i$ and $\overline \rho_i$ differing from $\del q$ and $(2-\del)q$, and the deviations from the entries that are fixed to $-q$.  We write $g(n) = f(n) + O(h(n))$ to denote two-sided error, i.e. $f(n) - C h(n) \le g(n) \le f(n) + C h(n)$.   


\begin{lem}
\label{OddCalcProp}
Let $ x$ and $y$ be unit vectors of dimension $n_1$ and $n_2$ respectively. Then
\begin{enumerate}
\item $\E[  u\cdot (M_i y) | E_{i-1}, H_{i-1}] = (\del -1) n_1 q ( v \cdot y)(1+o(1)) + O( n_1 n_2 p^2 \| y\|_\infty ) $ 
\item $\var[u \cdot (M_iy)| E_{i-1}, H_{i-1} ] = n_1 q(1+o(1))  + O(n_1 n_2 p^2 \| y \|_\infty^2 ) $.
\item $\E [  v \cdot (M_i^T x)| E_{i-1}, H_{i-1}] = (\del -1) n_2 q ( u \cdot x)(1+o(1))+ O(n_2 p \|x \| _{\infty})$ 
\item  $\var[ v \cdot (M_i^T x) | E_{i-1}, H_{i-1}] =n_2 q (1+o(1)) +O(n_2 p^2 \|x \|_{\infty}^2)$.
\item $\E [ \|M_i y\|_2^2 |E_{i-1}, H_{i-1}] = n_1 q (1+o(1)) + (\del -1)^2 n_1 q^2 (  v \cdot  y)^2(1+o(1))  +O( n_1 n_2^2 p^4 \|y \|_{\infty}^2 )    $.
\item $\E [ \|M_i^T x\|_2^2 | E_{i-1}, H_{i-1}] = n_2 q (1+o(1)) +  (\del -1)^2 n_2 q^2 ( u \cdot  x)^2(1+o(1))  +O(n_2p^2 \|x \|_\infty^2)      $.
\item $\var[ \|M_i y\|_2^2 | E_{i-1}, H_{i-1}] = O( n_1q(   \| y\|_1 \cdot  \|y \|_\infty^3 + q^3 \|y \|_1^4 + q \|y\|_\infty^2  \cdot \|y\|_1^2 + q + q^2 \|y\|_1^2) ) $.
\item $\var [ \|M_i^Tx \|_2^2 | E_{i-1}, H_{i-1}] = O(n_2n_1 q \|x\|_\infty^4)  $.
\end{enumerate}
\end{lem}


\begin{proof}
   If $j \in A_1$,
\begin{align*}
\E [(M_i y)_j |E_{i-1}, H_{i-1} ] &= - \sum_{l=1}^{n_2} q y_l  + \sum_{\substack {l \in A_2\\ (j,l) \notin E_{i-1} }} \rho_i y_l    + \sum_{ \substack{l \in B_2 \\(j,l) \notin E_{i-1} } } \overline \rho_i y_l \\
&= (\del -1) q ( v \cdot y)(1+O(p)) + O(n_2 p q \| y \|_{\infty}   ) \\
\var(My)_j &=\sum_{\substack {l \in A_2\\ (j,l) \notin E_{i-1} }} \rho_i (1 -\rho_i) y_l^2    + \sum_{ \substack{l \in B_2 \\(j,l) \notin E_{i-1} } } \overline \rho_i (1 -\overline \rho_i)  y_l^2
\end{align*}
and similarly for $j \in B_1$.

 This gives
 \begin{align*}
\E [ u \cdot (M_iy) |E_{i-1}, H_{i-1}] &= (\del -1 ) n_1 q  (v \cdot y)(1+O(p))  + O( n_1 n_2 p^2 \| y\|_\infty )\\
\var[ u \cdot (M_iy) |E_{i-1}, H_{i-1} ] &= n_1 q \| y \|_2^2  +O(n_1 q^2 \| y \|_2^2 + n_1 n_2 p^2 \| y \|_\infty^2 ) \\
&= n_1 q \| y \|_2^2 (1+O(p))  + O(n_1 n_2 p^2 \| y \|_\infty^2 ) 
\end{align*}

Then if $j \in A_2$,
\begin{align*}
\E [ (M_i^Tx)_j |E_{i-1}, H_{i-1}] &= - \sum_{l=1}^{n_1} q x_l  + \sum_{\substack{l \in A_1  \\ (l,j) \notin E_{i-1}}} \rho_i x_l    + \sum_{\substack{l \in B_1 \\ (l,j) \notin E_{i-1}  } } \overline \rho_i x_l \\
&= (\del -1) q (u \cdot x)(1+O(p)) + O( p \|x\|_{\infty} )  \\
\var [(M_i^Tx)_j |E_{i-1}, H_{i-1}] &=\sum_{\substack{l \in A_1  \\ (l,j) \notin E_{i-1}}} \rho_i (1 -\rho_i) x_l^2    + \sum_{\substack{l \in B_1 \\ (l,j) \notin E_{i-1}  } } \overline \rho_i (1 -\overline \rho_i)  x_l^2
\end{align*}
and similarly for $j \in B_2$.

   This gives
 \begin{align*}
\E [ v \cdot (M_i^Tx) |E_{i-1}, H_{i-1}]&= (\del -1 ) n_2 q  (u \cdot x)(1+O(p)) +O(n_2 p \|x \| _{\infty})  \\
\var[  v \cdot (M_i^Tx) |E_{i-1}, H_{i-1}] &= n_2 q  \| x \|_2^2 (1+ O(p)) + O(n_2 p^2 \|x \|_{\infty}^2)
\end{align*}

Finally we have

\begin{align*}
\E [ \|M_i y \|_2^2 |E_{i-1}, H_{i-1} ] &=  n_1 q \| y \|_2^2(1+ O(p)) +  (\del -1)^2 q^2 n_1 ( v \cdot y)^2 (1+O(p)) + O( n_1 n_2^2 p^4 \|y \|_{\infty}^2 )
\end{align*}
and
\begin{align*}
\var [ \|M_i y \|_2^2 |E_{i-1}, H_{i-1}] &= \sum_{i=1}^{n_1} \var ( (M y)_i )^2 ) \\
&\le \sum_{i=1}^{n_1}  \E  (  (M y)_i )^4  ) \\
&= O(  n_1 ( \|y\|_\infty^3 \|y\|_1 q + q^4 \|y\|_1^4 + q^2 \|y\|_2^4 + q^3 \|y\|_1^2 \|y\|_2^2 + q^2 \|y\|_\infty^2 \|y\|_1^2   )) \\
&= O( n_1 q (   \|y\|_\infty^3 \|y\|_1+q^3 \|y\|_1^4 + q  + q^2 \|y\|_1^2  + q \|y\|_\infty^2 \|y\|_1^2    ))
\end{align*}
and
  \begin{align*}
\E [ \|M_i^T x \|_2^2 |E_{i-1}, H_{i-1} ] &=  n_2 q  \| x \|_2^2(1+ O(p))
 +  (\del -1)^2 q^2 n_2 ( u\cdot x)^2(1+O(p)) + O(n_2p^2 \|x \|_\infty^2)
\end{align*}
and
\begin{align*}
\var [ \|M_i^T x \|_2^2 |E_{i-1}, H_{i-1}] &= \sum_{i=1}^{n_2} \var ( (M^T x)_i )^2 ) \\
&\le \sum_{i=1}^{n_2}  \E  (  (M^T x)_i )^4  ) \\
&=  n_2 \|x \|_\infty^4 \cdot O \left ( n_1 q  +n_1^2 q^2 + n_1^3 q ^3 + n_1 ^4 q^4    \right ) \\
&= O \left(   n_2 n_1 q  \|x \|_\infty^4 \right)
\end{align*}
\end{proof}

For cleaner notation in the rest of the proof we will write simply $\E [ \cdot ] $ for $\E [\cdot|E_{i-1}, H_{i-1} ]$ when working with the matrix $M_i$. 

Next we show the normalizing factors $\| M_i y \|_2$ and $\|M_i^T x \|_2$ are concentrated at each step;  the $l_\infty$ norms of the $x^i$'s are bounded over all iterations, and the $l_\infty$ and $l_1$ norms of the $y^i$'s are bounded.  This proposition is critical in ensuring steady progress of our potential functions.

\begin{lem}
\label{OddNormalizeProp}
With probability $1-O \left (T n_1^{-1/6} \right )$, for all $i=1, \dots T$,
\begin{enumerate}

\item $\|M_i y^i\|_2^2 = (n_1 q \|y^i \|_2^2 + (\del -1)^2 n_1 q^2 (  v \cdot  y^i)^2)  (1+ o(1) )$
\item $ \| M_i^T x^i \|_2^2 = (n_2 q \|  x^i \|_2^2 +  (\del -1)^2  n_2 q^2 ( u \cdot  x^i)^2)  (1+ o(1)) $
\item $\|x^i \|_\infty \le n_1^{-1/3} $
\item $\|y^i \|_\infty \le n_2^{-1/4} n_1^{-1/12} $
\item $\| y^i \|_1 \le  4 \sqrt { n_2 n_1 q}$
\end{enumerate}
\end{lem}

\begin{proof}
We begin by showing that
\begin{equation}\label{111}
\left | \left \{ j: |y^i_j| > \sqrt{ \frac{2 q n_1  }{n_2}}  \right \} \right | \le 3 n_2 n_1 q.
\end{equation}

We bound the number $L$ of $(1-q)$ entries  in $M_{i-1}$. $L$ is stochastically bounded by a $Binom(n_2 n_1, 2q)$ random variable, and so,

\begin{align*}
\Pr[ L \ge 3 n_2 n_1 q ] &\le e^{-qn_2n_1} = e^ {- \Theta(\sqrt {n_2 n_1}) }.
\end{align*}

The remaining entries have value $-q$.  If the $j$th row of $M_{i-1}$ has only $-q$ entries, then
\begin{align*}
|y^i_j| &\le \frac{ q \|x^{i-1}\|_1 }{\sqrt {n_2 q/2}} \\
&\le \sqrt{ \frac{2 q n_1  }{n_2}} \\
\end{align*}
using (2) inductively.   This proves (\ref{111}).

To prove (5), partition the coordinates of $y^i$ into two sets $\Delta$ and $\overline \Delta$, with $\Delta$ corresponding to rows of $M_{i-1}$ with every entry $-q$, and $\overline \Delta$ the rest.  Then
\begin{align*}
 \|y^i\|_1 & \le \sum_{j \in \Delta} |y^i_j| + \sum_{j \in \overline \Delta} |y^i_j| \\
 &\le \sqrt{ \frac{2 q n_1  }{n_2}} |\Delta| + \sqrt{ | \overline \Delta |}  \quad \text{ using part (2) inductively } \\
 &\le \sqrt{2  n_2 n_1 q} +\sqrt { 3 n_2 n_1 q} \\
 &\le  4 \sqrt {q n_2 n_1}
\end{align*}

We show by induction that whp the following hold for $i=1, \dots T$:
\begin{enumerate}
\item $\|M_i y^i\|_2^2 = \E \left( \|M_i y^i\|_2^2  \right ) (1 +O(n_1^{-1/8}) )$
\item $\|M_i^T x^i\|_2^2 = \E \left( \|M_i^T x^i\|_2^2  \right ) (1 +O(n_1^{-1/12}) )$
\item $\|x^i \|_\infty \le n_1^{-1/3} $
\item $\|y^i \|_\infty \le n_2^{-1/4} n_1^{-1/12} $
\end{enumerate}

Conditional on $y^i$ and $x^i$ respectively, we have
\begin{align*}
\E \left [  \|M_i y^i\|_2^2 \right ]   &= n_1q + (\del-1)^2 n_1q^2 ( v \cdot y^i)^2 + O(n_1q^2)   \\
\E  \left[ \|M_i^T x^i\|_2^2  \right ]  &= n_2q + (\del-1)^2 n_2q^2 (u \cdot x^i)^2 + O(n_2q^2)
\end{align*}

Using Chebyshev and part (3),
\begin{align*}
&\Pr \left[ \big | \|M_i y^i\|_2^2 - \E \left (  \|M_i y^i\|_2^2 \right) \big | > \alpha  \E \left (  \|M_i y^i\|_2^2 \right)   \right ]\\
 &\le \frac{ \var(\|M_i y^i\|_2^2   )   }{ n_1^2q^2  \alpha^2    } \\
&= \alpha^{-2} \cdot O \left ( \frac{ \| y\|_1 \cdot  \|y \|_\infty^3 + q^3 \|y \|_1^4 + q \|y\|_\infty^2  \cdot \|y\|_1^2 + q + q^2 \|y\|_1^2  }{n_1 q  }   \right)  \\
&= \alpha^{-2} \cdot O \left ( \frac{q^{1/2}  n_2^{-1/4} n_1^{1/4} + q^5n_2^2n_1^2 + q^2n_2^{1/2} n_1^{5/6}  + q + q^3 n_2 n_1  }{n_1 q  }   \right)  \\
&=  \alpha^{-2} \cdot O \left( q^{-1/2}  n_2^{-1/4} n_1^{-3/4} + q^4n_2^2n_1 + qn_2^{1/2} n_1^{-1/6}  + n_1^{-1} + q^2 n_2    \right ) \\
&= \alpha^{-2} \cdot O \left(  n_1^{-1/2} +n_1^{-1} + n_1^{-2/3}  + n_1^{-1} + n_1^{-1}    \right ) \\
&= O\left (  \frac{1}{  n_1^{1/2} \alpha^2}  \right) \\
&= O\left (     n_1^{-1/4} \right ) \quad \text{ for } \alpha = n_1^{-1/8}.
\end{align*}

Similarly, using Chebyshev and part (4),
\begin{align*}
&\Pr \left[ \big | \|M_i^T x^i\|_2^2 - \E \left (  \|M_i^T x^i\|_2^2 \right) \big | > \alpha  \E \left (  \|M_i^T x^i\|_2^2 \right)   \right ]\\
 &\le \frac{ \var(\|M_i^T x^i\|_2^2   )   }{ n_2^2q^2  \alpha^2    } \\
&= O \left( \frac{ n_1 \|x^i\|_\infty^4  }{ n_2 q \alpha^2  }  \right ) \\
&= O \left(\frac{ 1 }{ n_2 n_1^{1/3}  q \alpha^2  }  \right ) \\
&= O \left( \frac{n_1^{1/6} }{ n_2^{1/2} \alpha^2  }  \right ) \\
&= O\left (     n_1^{-1/6} \right ) \quad \text{ for } \alpha = n_1^{-1/12}.
\end{align*}

To prove (3), note that
\begin{align*}
\|x^{i+1}\|_\infty = \max_{j \in [n_1]} \frac{|(M^T_iy^i)_j|   }{ \|M^T_iy^i\|_2  }
\end{align*}

Using part (1), $\|M^T_i y^i\|_2 \ge \sqrt {n_1 q}/2 $ with probability $1- O(n_1^{-1/4})$.  Therefore it suffices to show that for every $j=1, \dots n_1$,
\[  |(M^T_iy^i)_j|  \le \frac{ n_1^{-1/3}  \sqrt {n_1 q} }{2} = \frac{ \sqrt d n_1^{1/6} n_2^{-1/4}  } { 2} . \]

To this end we will show that for any $j$,
\begin{equation}
\label{eq:OddXiprobBound}
\Pr \left[  |(M^T_iy^i)_j|  > \frac{ \sqrt d n_1^{1/6} n_2^{-1/4}  } { 2} \right ] \le \frac{1}{n_1^2}
\end{equation}

Again partition the coordinates of $y^i$, with $\Delta $ being the set of $j$ so that $|y^i_j| \le \sqrt{ \frac{2 q n_1  }{n_2}}$ and $\overline \Delta $ the rest.  The contribution to $ |(M^T_iy^i)_j|$ from $\Delta$ is bounded by
\begin{align*}
\left(n_2 q  + m_j \right)\sqrt{ \frac{2 q n_1  }{n_2}}
\end{align*}
where $m_j$ is the number of $1-q$ entries in the $j$th row of $M^T_i$. This number $m_j$ is dominated by a $Binom(n_2, 2q)$ random variable and so with probability $1 - \exp ( -n_2q)$, $m_j \le 3 n_2 q$. Therefore, the contribution from $\Delta$ is bounded by
\begin{align*}
\left( n_2 q + 3n_2q \right)\sqrt{ \frac{2 q n_1  }{n_2}} &\le 5 \sqrt {  n_2 n_1 q^3} \\
 &= O\left ( (n_2 n_1)^{-1/4}  \right )  = o( n_1^{1/6} n_2 ^{-1/4}).
\end{align*}
The contribution to $|(M^T_iy^i)_j|$ from $\overline \Delta$ is bounded by
\begin{align*}
\left(3 n_2 n_1 q \cdot q  +   \overline  m_j \cdot 1\right) n_2^{-1/4} n_1^{-1/12}
\end{align*}
where we have used (4) and (\ref{111})), and $\overline m_j$ is the number of $1-q$ entries in the $j$th row of $M^T_i$ whose column has index in $\overline \Delta$.  $\overline m_j$ is dominated by a $Binom( 3 n_2 n_1 q, q)$ random variable, and so with probability $1- O(\exp( -\Omega(3 n_2 n_1{13/12} q^2)))$, $\overline m_j \le 3 n_2 n_1 q^2 \cdot n_1^{1/12}$ in which case we have that the contribution from $\overline \Delta$ is bounded by
\begin{align*}
&3 n_2 n_1 q^2 n_2^{-1/4} n_1^{-1/12} + 3 n_2 n_1 q^2 n_2^{-1/4} \\
&= 3 d^2 n_2^{-1/4} n_1^{-1/12} + 3 d^2 n_2^{-1/4}\\
&\le \frac{ \sqrt d n_1^{1/6} n_2^{-1/4}  } { 4}
\end{align*}
proving inequality (\ref{eq:OddXiprobBound}). (We remark that for this part, the loose bounds we have above suffice; it is the next part that controls parameter settings).

To prove (4), set $\lam = n_2 ^{-1/4} n_1^{-1/12}$.
\begin{align*}
\|y^{i+1}\|_\infty = \max_{j \in [n_2]} \frac{|(M_i^T x^i)_j|   }{ \|M_i^T x^i\|_2  }
\end{align*}

Using part (2), $\|M_i^T x^i\|_2 \ge \sqrt {n_2 q}/2 $ with probability $1- O(n_1^{-1/6})$.  Therefore it suffices to show that for every $j=1, \dots n_2$,
\begin{align*}
 |(M_i^T x^i)_j|  &\le \frac{ \lam  \sqrt {n_2 q} }{2}  \\
 &=\frac{\lam \sqrt {d} n_2^{1/4}  }{ 2 n_1^{1/4}}
\end{align*}

We will show that for any $j$,
\begin{equation}
\label{eq:OddYiprobBound}
\Pr \left[  |(M_i^T x^i)_j|  > \frac{\lam \sqrt {d} n_2^{1/4}  }{ 2 n_1^{1/4}}  \right ] \le \frac{1}{n_2^2}
\end{equation}

We partition the coordinates of $x^i$ according to their magnitude, in bins $B_1, \dots B_L$, defined for $l <L$ as
\[
B_l = \left\{ i \, : \, |x_i| \in \left ( \frac{ n_1^{-1/3} }{2^l } , \frac{n_1^{-1/3}}{2^{l-1}}  \right ] \right\}  \]
with the interval for $B_L$ being $ [0, n_1^{-1/3}/2^{L-1}]$.  We set $L = \lceil \log ( n_1^{1/6}) \rceil$. Let
\[ t_l = | B_l| \le 2^{2l} n_1^{2/3} \]
using the fact that $x^i$ has unit $2$-norm.

We will bound the probability that bin $l$ contributes more than $\beta_l$ towards the value of $|(M_i^T x^i)_j|$, with
\[ \beta_l =  \frac{\lam \sqrt {d} n_2^{1/4}  }{ 4 n_1^{1/4} l^2}  \]
If all bins fall within these bounds, then
\begin{align*}
|(M_i^T x^i)_j| \le \sum_l \beta _l \le  \frac{\lam \sqrt {d} n_2^{1/4}  }{ 2 n_1^{1/4} }
\end{align*}
and therefore $\| y^{i+1}\|_\infty \le  n_2^{-1/4} n_1^{-1/12} $.

Let $Z_l \sim \text{Binom}(t_l,q)$. The contribution of  bin $l$ is bounded by the maximum of $  \frac{n_1^{-1/3}}{2^{l-1}} Z_l$  and $\left | q \sum_r x^i_r \right | \le q \sqrt {n_1}  \le \beta_l$.  To bound the first term, let
\begin{align*}
m_l &=   \beta_l 2^{l-1} n_1^{1/3} \\
&= \frac{\lam 2^l \sqrt {d} n_2^{1/4} n_1^{1/12}  }{ 8 l^2} \\
&= \frac{ 2^l \sqrt {d}  }{ 8 l^2}
\end{align*}
and consider
\begin{align*}
\Pr \left [ Z_l \ge m_l  \right ] &\le 2 \binom {t_l}{m_l} q^{m_l}   \\
&\le 2  \left( \frac{e t_l q }{m_l }  \right )^{m_l} \\
&\le 2  \left( \frac{e  2^{2l} n_1^{2/3} d (n_2 n_1)^{-1/2}  }{ \sqrt d 2^{l} / (8 l^2)}  \right )^{ \sqrt d 2^{l} / (8 l^2)} \\
&= 2  \left( \frac{8e l^2  2^{l} \sqrt d n_1^{1/6}  }{   n_2^{1/2}  }  \right )^{ \sqrt d 2^{l} / (8 l^2)} \\
&\le 2  \left( \frac{8e \sqrt d \log^2 (n_1)    n_1^{1/3}  }{   n_2^{1/2}  }  \right )^{ \sqrt d 2^{l} / (8 l^2)} \\
&\le 2  \left(  \frac{8e \sqrt d \log^2 (n_1)      }{ n_2^{1/6}  }   \right ) ^{\sqrt d /4} \\
&\le \frac{(8e  \sqrt d \log^2 (n_1)   )^7   }{n_2^{7/3} } \quad \text{ for } \sqrt d /4 \ge 7.
\end{align*}
   Taking a union bound over all $L$ bins, we have (\ref{eq:OddYiprobBound}).

\end{proof}

Next we show that the vector $y^i$ reaches high correlation with $ v$ after $T/2 $ steps.  Recall the definitions $V_i :=  v \cdot y^i$ and $U_i:= u \cdot x^i$.

\begin{prop}
\label{xiyievolve}
With probability $1-O(( \ln n_1)^{-2})$,  one of the following happens:
\begin{enumerate}
\item  For all $i \in \{T/2  , \dots  T\}$,
\[ V_i \ge \frac{(\del -1) \sqrt { n_2 n_1 q}}{4} \]
\item For all $l \in \{T/2 , \dots T \}$,
\[ V_i \le - \frac{(\del-1) \sqrt { n_2 n_1 q}}{4} \]
\end{enumerate}

\end{prop}

First we need the following bounds on the progress of $U_i$:

\begin{prop}
\label{chebPropOdd}
The following bounds on $U_i$ hold:
\begin{enumerate}

\item With probability at least $1/2$, $|U_{i}| \ge 1/4$ regardless of the value of $V_i$.

\item If $1/4 \le |U_{i}| \le  \ln n_1$, then with probability at least $1/2$, $|U_{i+1}| \ge 2 |U_i|$.
\item $ \Pr \left [| U_{i+1}| \ge 2 |  U_i |  \right ]  \ge 1 - \frac{1}{| U_i|^2}  $ for $ \ln n_1 \le | U_{i} | \le \sqrt {n_1} /4$.
\item If  $  U_{i}  \ge \sqrt {n_1} /4$, then  $\Pr \left [  U_{i+1} \ge \sqrt {n_1} /2  \right ]  \ge 1- O(1/\sqrt{n_1 n_2})$. Similarly, if $  U_{i}  \le - \sqrt {n_1} /4$, then  $\Pr \left [ U_{i+1} \le - \sqrt {n_1} /2  \right ]  \ge 1- O(1/\sqrt{n_1 n_2})$.
\item If $U_i \ge \sqrt {n_1}/2$, then $V_{i+1} \ge \frac{(\del -1) \sqrt { n_2 n_1 q}}{4}$ with probability $1-O(1/\sqrt{n_1 n_2})$.
\end{enumerate}
\end{prop}

1) and 2) ensure that Phase 1 succeeds, and that $U_i$ attaints value $ \ln n_1$ within $\ln n_1$ steps. 3) and 4) ensure steady progress in Phase 2 and that once $U_i$ attains a high value, it maintains it. 5) connects the two potential functions by showing that $V_{i+1}$ is large if $U_i$ is large.

\begin{proof}[Proof of Proposition \ref{chebPropOdd}]

1. The variance of $u \cdot (M_i y^i)$ is $ \sim n_1q$, and so  a Berry-Esseen bound gives that with probability at least $1/2$, $|u \cdot M_i y^i| \ge \sqrt{ n_1 q}/4$.  Then using Lemma \ref{OddNormalizeProp}, we have that $||M_i y^i||_2 = \sqrt {n_1 q}(1+o(1))$ whp, and so with probability at least $1/2$, $|U_{i}|=|u \cdot x^{i}| \ge 1/4$.

2.   We prove this in two steps. The expectation of $ v \cdot (M_i^T x^i)$ is $(\del -1)n_2 q( u \cdot x^i)$, with variance $n_2 q$.  Both are $\omega(1)$, and the expectation is at least $(\del - 1)/4$ times the variance in absolute value, and so whp, $v \cdot (M_i^T x^i) = (\del -1)n_2 q(u \cdot x^i) (1+o(1))$. Using Lemma \ref{OddNormalizeProp} again, we have that whp, $V_{i+1} = (\del -1) \sqrt {n_2 q} (u \cdot x^i)$.

Conditioning on this value, we have
\[
\E[u \cdot (M_{i+1} y^{i+1})] = (\del-1)^2 \sqrt{n_2 q} n_1 q (u \cdot x^i)(1+o(1))
\]
and its variance is $n_1 q$.  With probability $1/2$ we have $|u \cdot (M_{i+1} y^{i+1})|\ge  (\del-1)^2 \sqrt{n_2 q} n_1 q (u \cdot x^i)(1-o(1))$, and then normalizing with Lemma \ref{OddNormalizeProp} we have $|U_{i+1}| \ge (\del-1)^2 \sqrt{n_2 n_1} q |U_i|$, which from our choice of $q$, is at least $2 | U_i|$.

3. Similar to the above.  Apply Chebyshev so that $v \cdot (M_i^T x^i) = (\del -1 )n_2 q (u \cdot x^i)(1+o(1))$ with probability $1- o(1)$, and normalize so that $v \cdot y^{i+1} = (\del -1)\sqrt {n_2 q} (u \cdot x^i) (1+o(1))$ whp.  Now the expectation of $u \cdot (M_{i+1} y^{i+1})$ is $(\del-1)^2 \sqrt{n_2 q} n_1 q (u \cdot x^i)(1+o(1))$ with variance $n_1 q$, and so applying Chebyshev, we have
\begin{align*}
\Pr [ |u \cdot (M_{i+1} y^{i+1})| < (\del-1)^2 \sqrt{n_2 q} n_1 q |u \cdot x^i|/2] &\le \frac{n_1 q}{ (\del-1)^4 n_2  n_1^2 q^3 (u \cdot x^i)^2/4} \\
&= \frac{4}{ (\del-1)^4 n_2  n_1 q^2 (u \cdot x^i)^2} \\
&\le \frac{1}{25 (u \cdot x^i)^2}
\end{align*}

Then normalizing, and using Lemma \ref{OddCalcProp} and part (2) above, we get
\begin{align*}
|U_{i+1}| & \ge \frac{ (\del-1)^2 \sqrt{n_2 q} n_1 q |U_i|   }{ 2 \sqrt{ n_1 q + (\del-1)^2 n_1 q^2 (V_{i+1})^2   }    } \\
& \ge \frac{ (\del-1)^2 \sqrt{n_2 q} n_1 q |U_i|   }{ 2 \sqrt{ n_1 q + (\del-1)^4 n_2 n_1 q^3 (U_i)^2   }    } \\
&\ge 2 | U_i|.
\end{align*}

4,5.  Chebyshev again.

\end{proof}

\begin{proof}[Proof of Proposition \ref{xiyievolve}]

  In the first phase, we show that it takes $ \ln n_1$ iterations for $|U_i|$ to reach $ \ln n_1$ whp.  Next, it takes a further $\ln n_1$ iterations to reach $\sqrt {n_1}/2$.  Finally, $|U_i|$ will remain above $\sqrt {n_1}/2$ whp for an additional $2 \ln n_1$ iterations.

Step 1: We call a step from $U_i$ to $U_{i+1}$ `good' if $|U_{i+1}| \ge 2 | U_i|$, or if $|U_{i+1}| \ge 1/4 $ following a `bad' step.  A run of $\ln \ln n_1$ good steps must end with $|U_i| \ge  \ln n_1$.  As long as $|U_i| <  \ln n_1$, the proposition above shows that the probability of a good step is at least $1/2$, so in $ \ln n_1$ steps, with probability $1-o(1)$ we will either have such a run of $\ln \ln n_1$ good steps or reach $ \ln n_1$ even earlier.

Step 2:  Once we have $|U_i| \ge  \ln n_1$, the value will double whp in successive steps until $|U_i| \ge \sqrt {n_1}/4$.  This takes at most $\ln n_1$ steps.  The total error probability, by part 3) of Proposition \ref{chebPropOdd} is a geometric series that sums to $O(1/ ( \ln n_1)^2)$.

Step 3:  Once $|U_i| \ge \sqrt{n_1}/{4}$ then for the next $2 \ln n_1$ steps, $U_{i+1}, U_{i+2}, \dots$, we have $|U_i| \ge \sqrt{n_1}/{2}$, with total error probability $O(T/\sqrt{n_1n_2})$.

Step 4: Finally we use  part 5) of Proposition \ref{chebPropOdd} to conclude that $y^i$ has high correlation with $v$.

\end{proof}

We now use Proposition \ref{xiyievolve} to prove the main theorem.

\begin{proof}[Proof of Theorem \ref{thm:block}]

Now that we know whp $y^{T/2 }, y^{T/2 +1}, \dots$ all have large correlation with $v$, we show that taking the majority vote for each coordinate of $z^{T/2 +1}, z^{T/2+2}, \dots$ recovers $\pm u$ whp.

Take the first case from Proposition \ref{xiyievolve}, with $V_i \ge \frac{(\del -1) \sqrt { n_2 n_1 q}}{4} $.  Assume $j \in A_1$, then we have, conditioned on the value of $V_i$

\begin{align*}
\Pr[z^{i+1}_j = 1] &= \Pr[ x_j > 0] \\
& \ge 1- \frac{\var((My^i)_j)}{ ( \E( (My^i)_j ) )^2 }  \\
& \ge 1 - \frac{ 32q}{ (\del-1)^4 q^3 n_1 n_2 } \\
&=   1 - \frac{ 32}{ 100^2 } \ge .9
\end{align*}
Now an application of Azuma's inequality shows that with probability at least $1- o(n_1^{-2})$, $\sum _{i = T/2}^T z^i_j >0$.  Similarly, for $j \in B_1$, we have $\sum _{i = T/2}^T z^i_j >0$ with probability at least $1-o(n_1^{-2})$, and so whp the majority vote recovers $u$ exactly.  The same argument shows that if the second case of Proposition \ref{xiyievolve} holds, then we find $-u$ whp.

\end{proof}

\begingroup
\raggedright
\sloppy
\bibliographystyle{plain}	
\bibliography{planted}

\begin{thebibliography}{10}

\bibitem{abbe2014exact}
Emmanuel Abbe, Afonso~S Bandeira, and Georgina Hall.
\newblock Exact recovery in the stochastic block model.
\newblock {\em arXiv preprint arXiv:1405.3267}, 2014.

\bibitem{achlioptas2005hiding}
Dimitris Achlioptas, Haixia Jia, and Cristopher Moore.
\newblock Hiding satisfying assignments: Two are better than one.
\newblock {\em J. Artif. Intell. Res.(JAIR)}, 24:623--639, 2005.

\bibitem{achlioptas2001fast}
Dimitris Achlioptas and Frank McSherry.
\newblock Fast computation of low rank matrix approximations.
\newblock In {\em Proceedings of the thirty-third annual ACM symposium on
  Theory of computing}, pages 611--618. ACM, 2001.

\bibitem{Alekhnovich11a}
Michael Alekhnovich.
\newblock More on average case vs approximation complexity.
\newblock {\em Computational Complexity}, 20(4):755--786, 2011.

\bibitem{alon1997spectral}
Noga Alon and Nabil Kahale.
\newblock A spectral technique for coloring random 3-colorable graphs.
\newblock {\em SIAM Journal on Computing}, 26(6):1733--1748, 1997.

\bibitem{applebaum2012pseudorandom}
Benny Applebaum.
\newblock Pseudorandom generators with long stretch and low locality from
  random local one-way functions.
\newblock In {\em Proceedings of the 44th symposium on Theory of Computing},
  pages 805--816. ACM, 2012.

\bibitem{applebaum2010public}
Benny Applebaum, Boaz Barak, and Avi Wigderson.
\newblock Public-key cryptography from different assumptions.
\newblock In {\em Proceedings of the 42nd ACM symposium on Theory of
  computing}, pages 171--180. ACM, 2010.

\bibitem{Barak:2012:Online}
Boaz Barak.
\newblock Truth vs proof: The unique games conjecture and {F}eige's hypothesis.
\newblock
  \url{http://windowsontheory.org/2012/07/31/truth-vs-proof-the-unique-games-conjecture-and-feiges-hypothesis/},
  July 2012.

\bibitem{barthel2002hiding}
Wolfgang Barthel, Alexander~K Hartmann, Michele Leone, Federico
  Ricci-Tersenghi, Martin Weigt, and Riccardo Zecchina.
\newblock Hiding solutions in random satisfiability problems: A statistical
  mechanics approach.
\newblock {\em Physical review letters}, 88(18):188701, 2002.

\bibitem{berthet2014optimal}
Quentin Berthet.
\newblock Optimal testing for planted satisfiability problems.
\newblock {\em arXiv preprint arXiv:1401.2205}, 2014.

\bibitem{BerthetR:13}
Quentin Berthet and Philippe Rigollet.
\newblock Complexity theoretic lower bounds for sparse principal component
  detection.
\newblock In {\em COLT}, pages 1046--1066, 2013.

\bibitem{Blum:92a}
Avrim Blum.
\newblock Learning boolean functions in an infinite attribute space.
\newblock {\em Machine Learning}, 9:373--386, 1992.

\bibitem{blum1995coloring}
Avrim Blum and Joel Spencer.
\newblock Coloring random and semi-random k-colorable graphs.
\newblock {\em Journal of Algorithms}, 19(2):204--234, 1995.

\bibitem{bogdanov2009security}
Andrej Bogdanov and Youming Qiao.
\newblock On the security of goldreich's one-way function.
\newblock In {\em Approximation, Randomization, and Combinatorial Optimization.
  Algorithms and Techniques}, pages 392--405. Springer, 2009.

\bibitem{boppana1987eigenvalues}
Ravi~B Boppana.
\newblock Eigenvalues and graph bisection: An average-case analysis.
\newblock In {\em Foundations of Computer Science, 1987., 28th Annual Symposium
  on}, pages 280--285. IEEE, 1987.

\bibitem{bottcher2005coloring}
Julia B{\"o}ttcher.
\newblock Coloring sparse random k-colorable graphs in polynomial expected
  time.
\newblock {\em Mathematical Foundations of Computer Science 2005}, page 156,
  2005.

\bibitem{CharikarWirth:04}
Moses Charikar and Anthony Wirth.
\newblock Maximizing quadratic programs: extending grothendieck's inequality.
\newblock In {\em Foundations of Computer Science, 2004. Proceedings. 45th
  Annual IEEE Symposium on}, pages 54--60. IEEE, 2004.

\bibitem{coja2006spectral}
Amin Coja-Oghlan.
\newblock A spectral heuristic for bisecting random graphs.
\newblock {\em Random Structures \& Algorithms}, 29:3:351--398, 2006.

\bibitem{coja2010graph}
Amin Coja-Oghlan.
\newblock Graph partitioning via adaptive spectral techniques.
\newblock {\em Combinatorics, Probability \& Computing}, 19(2):227, 2010.

\bibitem{coja2010efficient}
Amin Coja-Oghlan, Colin Cooper, and Alan Frieze.
\newblock An efficient sparse regularity concept.
\newblock {\em SIAM Journal on Discrete Mathematics}, 23(4):2000--2034, 2010.

\bibitem{coja2004strong}
Amin Coja-Oghlan, Andreas Goerdt, and Andr{\'e} Lanka.
\newblock Strong refutation heuristics for random k-sat.
\newblock In {\em Approximation, Randomization, and Combinatorial Optimization.
  Algorithms and Techniques}, pages 310--321. Springer, 2004.

\bibitem{coja2003certifying}
Amin Coja-Oghlan, Andreas Goerdt, Andr{\'e} Lanka, and Frank Sch{\"a}dlich.
\newblock Certifying unsatisfiability of random 2k-sat formulas using
  approximation techniques.
\newblock In {\em Fundamentals of Computation Theory}, pages 15--26. Springer,
  2003.

\bibitem{coja2004techniques}
Amin Coja-Oghlan, Andreas Goerdt, Andr{\'e} Lanka, and Frank Sch{\"a}dlich.
\newblock Techniques from combinatorial approximation algorithms yield
  efficient algorithms for random 2k-sat.
\newblock {\em Theoretical Computer Science}, 329(1):1--45, 2004.

\bibitem{coja2009finding}
Amin Coja-Oghlan and Andr{\'e} Lanka.
\newblock Finding planted partitions in random graphs with general degree
  distributions.
\newblock {\em SIAM Journal on Discrete Mathematics}, 23(4):1682--1714, 2009.

\bibitem{condon2001algorithms}
Anne Condon and Richard~M Karp.
\newblock Algorithms for graph partitioning on the planted partition model.
\newblock {\em Random Structures \& Algorithms}, 18(2):116--140, 2001.

\bibitem{DanielyLS:13}
Amit Daniely, Nati Linial, and Shai Shalev-Shwartz.
\newblock More data speeds up training time in learning halfspaces over sparse
  vectors.
\newblock In {\em NIPS}, pages 145--153, 2013.

\bibitem{DanielyS14}
Amit Daniely and Shai Shalev{-}Shwartz.
\newblock Complexity theoretic limitations on learning dnf's.
\newblock {\em CoRR}, abs/1404.3378, 2014.

\bibitem{DecaturGR99}
S.~Decatur, O.~Goldreich, and D.~Ron.
\newblock Computational sample complexity.
\newblock {\em SIAM Journal on Computing}, 29(3):854--879, 1999.

\bibitem{decelle2011asymptotic}
Aurelien Decelle, Florent Krzakala, Cristopher Moore, and Lenka Zdeborov{\'a}.
\newblock Asymptotic analysis of the stochastic block model for modular
  networks and its algorithmic applications.
\newblock {\em Physical Review E}, 84(6):066106, 2011.

\bibitem{dyer1989solution}
Martin~E. Dyer and Alan~M. Frieze.
\newblock The solution of some random np-hard problems in polynomial expected
  time.
\newblock {\em Journal of Algorithms}, 10(4):451--489, 1989.

\bibitem{feige2004easily}
Uriel Feige and Eran Ofek.
\newblock Easily refutable subformulas of large random 3cnf formulas.
\newblock In {\em Automata, languages and programming}, pages 519--530.
  Springer, 2004.

\bibitem{feige2005spectral}
Uriel Feige and Eran Ofek.
\newblock Spectral techniques applied to sparse random graphs.
\newblock {\em Random Structures \& Algorithms}, 27(2):251--275, 2005.

\bibitem{Feldman:07jmlr}
V.~Feldman.
\newblock Attribute efficient and non-adaptive learning of parities and {DNF}
  expressions.
\newblock {\em Journal of Machine Learning Research}, (8):1431--1460, 2007.

\bibitem{feldman2013statistical}
Vitaly Feldman, Elena Grigorescu, Lev Reyzin, Santosh Vempala, and Ying Xiao.
\newblock Statistical algorithms and a lower bound for planted clique.
\newblock In {\em Proceedings of the 45th annual ACM symposium on Symposium on
  theory of computing}, pages 655--664. ACM, 2013.

\bibitem{feldman2013complexity}
Vitaly Feldman, Will Perkins, and Santosh Vempala.
\newblock On the complexity of random satisfiability problems with planted
  solutions.
\newblock {\em arXiv preprint arXiv:1311.4821}, 2013.

\bibitem{flaxman2003spectral}
Abraham Flaxman.
\newblock A spectral technique for random satisfiable 3cnf formulas.
\newblock In {\em Proceedings of the fourteenth annual ACM-SIAM symposium on
  Discrete algorithms}, pages 357--363. Society for Industrial and Applied
  Mathematics, 2003.

\bibitem{FlorescuP:15manu}
Laura Florescu and Will Perkins.
\newblock Spectral thresholds in the bipartite stochastic block model.
\newblock preprint, 2015.

\bibitem{fredman1984storing}
Michael~L Fredman, J{\'a}nos Koml{\'o}s, and Endre Szemer{\'e}di.
\newblock Storing a sparse table with 0 (1) worst case access time.
\newblock {\em Journal of the ACM (JACM)}, 31(3):538--544, 1984.

\bibitem{friedman2005recognizing}
Joel Friedman, Andreas Goerdt, and Michael Krivelevich.
\newblock Recognizing more unsatisfiable random k-sat instances efficiently.
\newblock {\em SIAM Journal on Computing}, 35(2):408--430, 2005.

\bibitem{goerdt2001efficient}
Andreas Goerdt and Michael Krivelevich.
\newblock Efficient recognition of random unsatisfiable k-sat instances by
  spectral methods.
\newblock In {\em STACS 2001}, pages 294--304. Springer, 2001.

\bibitem{goerdt2003recognizing}
Andreas Goerdt and Andr{\'e} Lanka.
\newblock Recognizing more random unsatisfiable 3-sat instances efficiently.
\newblock {\em Electronic Notes in Discrete Mathematics}, 16:21--46, 2003.

\bibitem{goldreich2000candidate}
Oded Goldreich.
\newblock Candidate one-way functions based on expander graphs.
\newblock {\em IACR Cryptology ePrint Archive}, 2000:63, 2000.

\bibitem{holland1983stochastic}
Paul~W Holland, Kathryn~Blackmond Laskey, and Samuel Leinhardt.
\newblock Stochastic blockmodels: First steps.
\newblock {\em Social networks}, 5(2):109--137, 1983.

\bibitem{ishai2008cryptography}
Yuval Ishai, Eyal Kushilevitz, Rafail Ostrovsky, and Amit Sahai.
\newblock Cryptography with constant computational overhead.
\newblock In {\em Proceedings of the 40th annual ACM symposium on Theory of
  computing}, pages 433--442. ACM, 2008.

\bibitem{jerrum1998metropolis}
Mark Jerrum and Gregory~B Sorkin.
\newblock The metropolis algorithm for graph bisection.
\newblock {\em Discrete Applied Mathematics}, 82(1):155--175, 1998.

\bibitem{jia2005generating}
Haixia Jia, Cristopher Moore, and Doug Strain.
\newblock Generating hard satisfiable formulas by hiding solutions deceptively.
\newblock In {\em PROCEEDINGS OF THE NATIONAL CONFERENCE ON ARTIFICIAL
  INTELLIGENCE}, volume~20, page 384. Menlo Park, CA; Cambridge, MA; London;
  AAAI Press; MIT Press; 1999, 2005.

\bibitem{kearns1998efficient}
Michael Kearns.
\newblock Efficient noise-tolerant learning from statistical queries.
\newblock {\em Journal of the ACM (JACM)}, 45(6):983--1006, 1998.

\bibitem{korada2011gossip}
Satish~Babu Korada, Andrea Montanari, and Sewoong Oh.
\newblock Gossip pca.
\newblock In {\em Proceedings of the ACM SIGMETRICS joint international
  conference on Measurement and modeling of computer systems}, pages 209--220.
  ACM, 2011.

\bibitem{krivelevich2006solving}
Michael Krivelevich and Dan Vilenchik.
\newblock Solving random satisfiable 3cnf formulas in expected polynomial time.
\newblock In {\em Proceedings of the seventeenth annual ACM-SIAM symposium on
  Discrete algorithm}, pages 454--463. ACM, 2006.

\bibitem{krzakala2014reweighted}
Florent Krzakala, Marc M{\'e}zard, and Lenka Zdeborov{\'a}.
\newblock Reweighted belief propagation and quiet planting for random k-sat.
\newblock {\em Journal on Satisfiability, Boolean Modeling and Computation},
  8:149--171, 2014.

\bibitem{krzakala2013spectral}
Florent Krzakala, Cristopher Moore, Elchanan Mossel, Joe Neeman, Allan Sly,
  Lenka Zdeborov{\'a}, and Pan Zhang.
\newblock Spectral redemption in clustering sparse networks.
\newblock {\em Proceedings of the National Academy of Sciences},
  110(52):20935--20940, 2013.

\bibitem{krzakala2009hiding}
Florent Krzakala and Lenka Zdeborov{\'a}.
\newblock Hiding quiet solutions in random constraint satisfaction problems.
\newblock {\em Physical review letters}, 102(23):238701, 2009.

\bibitem{lin2010power}
Frank Lin and William~W Cohen.
\newblock Power iteration clustering.
\newblock In {\em Proceedings of the 27th International Conference on Machine
  Learning (ICML-10)}, pages 655--662, 2010.

\bibitem{massoulie2014community}
Laurent Massouli{\'e}.
\newblock Community detection thresholds and the weak ramanujan property.
\newblock In {\em STOC 2014: 46th Annual Symposium on the Theory of Computing},
  pages 1--10, 2014.

\bibitem{mcsherry2001spectral}
Frank McSherry.
\newblock Spectral partitioning of random graphs.
\newblock In {\em Foundations of Computer Science, 2001. Proceedings. 42nd IEEE
  Symposium on}, pages 529--537. IEEE, 2001.

\bibitem{mossel2013proof}
Elchanan Mossel, Joe Neeman, and Allan Sly.
\newblock A proof of the block model threshold conjecture.
\newblock {\em arXiv preprint arXiv:1311.4115}, 2013.

\bibitem{mossel2014consistency}
Elchanan Mossel, Joe Neeman, and Allan Sly.
\newblock Consistency thresholds for binary symmetric block models.
\newblock {\em arXiv preprint arXiv:1407.1591}, 2014.

\bibitem{mossel2012stochastic}
Elchanan Mossel, Joe Neeman, and Allan Sly.
\newblock Reconstruction and estimation in the planted partition model.
\newblock {\em Probability Theory and Related Fields}, pages 1--31, 2014.

\bibitem{nadakuditi2012graph}
Raj~Rao Nadakuditi and Mark~EJ Newman.
\newblock Graph spectra and the detectability of community structure in
  networks.
\newblock {\em Physical review letters}, 108(18):188701, 2012.

\bibitem{o2013goldreich}
Ryan O'Donnell and David Witmer.
\newblock Goldreich's prg: Evidence for near-optimal polynomial stretch.
\newblock In {\em Conference on Computational Complexity}, 2014.

\bibitem{Servedio:00jcss}
R.~Servedio.
\newblock Computational sample complexity and attribute-efficient learning.
\newblock {\em Journal of Computer and System Sciences}, 60(1):161--178, 2000.

\bibitem{Shalev-ShwartzST:12}
Shai Shalev-Shwartz, Ohad Shamir, and Eran Tromer.
\newblock Using more data to speed-up training time.
\newblock In {\em AISTATS}, pages 1019--1027, 2012.

\bibitem{Trevisan:2007:Online}
Luca Trevisan.
\newblock More ways to prove unsatisfiability of random k-sat.
\newblock
  \url{http://lucatrevisan.wordpress.com/2007/08/21/more-ways-to-prove-unsatisfiability-of-random-k-sat/},
  August 2007.

\bibitem{vu2014simple}
Van Vu.
\newblock A simple svd algorithm for finding hidden partitions.
\newblock {\em arXiv preprint arXiv:1404.3918}, 2014.

\bibitem{zhou2004clustering}
Hanson Zhou and David Woodruff.
\newblock Clustering via matrix powering.
\newblock In {\em Proceedings of the twenty-third ACM SIGMOD-SIGACT-SIGART
  symposium on Principles of database systems}, pages 136--142. ACM, 2004.

\end{thebibliography}
\endgroup


\end{document}